\documentclass[a4paper,english,10pt]{amsart} 
\usepackage{amsmath}
\usepackage{amsthm}
\usepackage{amsfonts}
\usepackage{amssymb}
\usepackage{graphicx}
\usepackage{xcolor}
\usepackage{enumerate}
\usepackage{csquotes}
\usepackage[english]{babel}
\usepackage{tikz}
\usepackage{fancyhdr}
\usepackage[
style=numeric,
backend=bibtex
]{biblatex}[=3.2]

\usepackage{siunitx}
\bibliography{main}

\theoremstyle{definition}
\newtheorem{definition}{Definition}[section]

\theoremstyle{remark}
\newtheorem{oss}[definition]{Remark}

\theoremstyle{plain}
\newtheorem{theorem}[definition]{Theorem}
\newtheorem{prop}[definition]{Proposition}
\newtheorem{lemma}[definition]{Lemma}

\DeclareMathOperator{\Sym}{Sym}
\DeclareMathOperator{\sign}{sign}
\DeclareMathOperator{\Aut}{Aut}

\def\A{\mathcal{A}}
\def\B{\mathcal{B}}
\def\M{\mathcal{M}}

\def\aux{\textup{aux}}
\def\BV{\mathrm{BV}}
\def\C{\mathbb{C}}

\def\ferm{\textup{ferm}}
\def\H{\mathcal{H}}
\def\Hoch{\mathrm{H}}

\def\phi{\varphi}

\def\t{\textup{t}}
\def\Z{\mathbb{Z}}

\newcommand{\siX}{{\small{$X_{0}$}}}
\newcommand{\siS}{{\small{$S_{0}$}}}
\newcommand{\siXS}{{\small{$(X_{0}, S_{0})$}}}

\newcommand{\swS}{{\small{$\widetilde{S}$}}}
\newcommand{\swX}{{\small{$\widetilde{X}$}}}
\newcommand{\swXS}{{\small{$(\widetilde{X}, \widetilde{S})$}}}
\newcommand{\SBV}{{\small{$S_{\BV}$}}}
\newcommand{\stS}{{\small{$S_{\t}$}}}
\newcommand{\stX}{{\small{$X_{\t}$}}}
\newcommand{\stXS}{{\small{$(X_{\t}, S_{\t})$}}}

\newcommand{\AI}{{\small{$\mathcal{A}_{0}$}}}

\newcommand{\isp}{{\small{$(\mathcal{A}_{0}, \mathcal{H}_{0},D_{0})$}}}
\newcommand{\ABV}{{\small{$\mathcal{A}_{\BV}$}}}
\newcommand{\HBV}{{\small{$\mathcal{H}_{\BV}$}}}
\newcommand{\DBV}{{\small{$D_{\BV}$}}}
\newcommand{\JBV}{{\small{$J_{\BV}$}}}
\newcommand{\BVsp}{{\small{$(\mathcal{A}_{\BV}, \mathcal{H}_{\BV},D_{\BV}, J_{\BV})$}}}

\newcommand{\ATOT}{{\small{$\mathcal{A}_{\t}$}}}
\newcommand{\HTOT}{{\small{$\mathcal{H}_{\t}$}}}
\newcommand{\DTOT}{{\small{$D_{\t}$}}}
\newcommand{\JTOT}{{\small{$J_{\t}$}}}
\newcommand{\TOTsp}{{\small{$(\mathcal{A}_{\t}, \mathcal{H}_{\t},D_{\t}, J_{\t})$}}}

\newcommand{\scB}{{\small{$\mathcal{B}$}}}
\newcommand{\scM}{{\small{$\mathcal{M}$}}}

\newcommand{\cBRST}{BV}

\newcommand{\BRST}{\mbox{\tiny{BRST}}}
\newcommand{\BVt}{\mbox{\tiny{BV}}}

\DeclareMathOperator{\tr}{Tr}

\title[The BV construction for finite spectral triples]{The BV construction for finite spectral triples}

\author{Roberta A. Iseppi}
\address{Georg-August-Universit\"{a}t G\"{o}ttingen, Mathematisches Institut, Bunsenstrasse 3-5, 37073 G\"{o}ttingen, Germany.}
\email{roberta.iseppi@mathematik.uni-goettingen.de}
\date{\today}

\begin{document}

\begin{abstract}
This article presents how the BV formalism naturally inserts in the framework of noncommutative geometry for gauge theories induced by finite spectral triples. Reaching this goal entails that not only all the steps of the BV construction, from the introduction of ghost/anti-ghost fields to the construction of the BRST complex, can be expressed using noncommutative geometric objects, but also that the method to go from one step in the construction to the next one has an intrinsically noncommutative geometric nature.\\
Moreover, we prove that both the classical BV and BRST complexes coincide with another cohomological theory, naturally appearing in noncommutative geometry: the Hochschild complex of a coalgebra. The construction is presented in detail for {\footnotesize{$\mathrm{U}(n)$}}-gauge theories induced by spectral triples on the algebra {\footnotesize{$M_{n}(\mathbb{C})$}}. 
\end{abstract}

\maketitle
\tableofcontents


\section{Why a noncommutative geometric approach to the BV formalism}
\label{Sect: intro}

\noindent
Since its early years, the introduction of the concept of quantization brought two new ideas into the framework of a  mathematical formalization of physics laws: {\emph{discreteness}} and {\emph{noncommutativity}}. Indeed, while the gravitational force, at macroscopic level, can be described as determined by the curvature of the spacetime and hence it has a {\emph{continuous}} nature within the context of Riemannian geometry, the other three fundamental interactions, that is the electromagnetic, strong and weak interactions, are {\emph{discrete}} quantum fields, whose interactions are mediated by elementary particles described by the Standard Model of particle physics.
Moreover, while Riemannian manifolds are intrinsically {\emph{commutative}}, the quantum observables have a {\emph{noncommutative}} character, not only when described as self-adjoint operators on a hilbert space, whose commutator depends on the Planck constant within the framework of the canonical quantization, but also when seen as elements in a noncommutative algebra with a star product obtained via a deformation quantization process. \\
\\
It was to include both these two aspects and hence being able to describe  the continuum as well as the discrete, the commutative as well as the noncommutative, that the classical geometrical notion of spacetime, given in terms of a Riemannian manifold within the context of differential geometry, was enlarged, arriving to the introduction of the notion of {\emph{spectral triple}}, in the context of noncommutative geometry. While a mathematically rigorous definition of the key notion of spectral triple will be given in Section \ref{Sect: intro to NCG and gauge theories} (cf. Definition \ref{def spectral triple}), a more intuitive approach describes a spectral triple as a noncommutative version of the classical notion of compact spin Riemannian manifold: the perspective changes and, while in the classical context the basic objects are topological spaces, their points and the coordinates charts used to describe them, the noncommutative geometric point of view sees the geometric space as encoded in its algebra of coordinates, which can then be taken to be also finite-dimensional or non-commutative. Hence, while the classical models can be recovered by considering commutative $\mathcal{C}^{*}$-algebras, by eliminating the requirement of them being commutative, we obtain an all new class of purely quantum models. Even more, very remarkably, this new geometrical object also has a deep relation to quantum field theory: indeed,  given a spectral triple, finite or infinite dimensional, commutative or not, it always encodes a {\emph{gauge theory}}, that is, a physical theory invariant under the action of a (local) group of symmetries. \\
\\
It is precisely this intrinsic relation existing between spectral triples and gauge theories that makes noncommutative geometry such an interesting and natural framework for a mathematical formalization and geometric interpretation of constructions which involve gauge theories, such as the Batalin-Vilkovisky (BV) construction. As proved in this article, both the classical and the quantum BV constructions for finite gauge theories naturally insert in the mathematical framework provided by noncommutative geometry. Moreover, these constructions are expected to naturally extend to the context of gauge theories induced by general spectral triples, defined on an infinite-dimensional algebra. \\
\\
In more detail, in Section \ref{Sect: intro to BV} we describe the classical and the quantum BV constructions in the context of finite gauge theories, from the introduction of ghost/anti-ghost fields to the construction of the induced quantum BRST cohomology complex while Section \ref{Sect: intro to NCG and gauge theories} is devoted to a brief recall of the key notions and constructions from noncommutative geometry.\\
After these more introductory parts, we prove how all the steps that compose the BV construction can be formalized in the language of spectral triples: in Section \ref{Sect: BV spectral triple} we introduce the notion of {\emph{BV spectral triple}}, which encodes in the structure of a real spectral triple both the extended configuration space as well as the extended action, which has to be a solution of the classical master equation. \\
As any BV-extended theory naturally induces a cohomology complex, namely the {\emph{BV complex}}, analogously it does the BV-spectral triple: in Section \ref{Sect: BV and Hochschild complex} we relate for the first time the BV complex to another cohomological theory which is more naturally embedded in the context of noncommutative geometry, that is, the Hochschild complex of a coalgebra over a comodule. Section \ref{Sect: gauge-fixing procedure in terms of NCG} is devoted to the {\emph{gauge-fixing process}} and explains how also the introduction of auxiliary fields, which are needed to define a gauge-fixing fermion, can be described within the setting of noncommutative geometry. In particular, the process of adding the auxiliary fields can be seen as a further enlargement of the BV-spectral triple. The result is the construction of a so-called {\em total spectral triple}, which contains all fields/ghost fields and auxiliary fields and such that the induced cohomology is quasi-isomorphic to the BV complex. \\
The {\emph{BRST complex}} is the central topic of Section \ref{Sect: BRST and Hochschild complex}, where we prove that the structure already found for the BV complex appears also in the gauge-fixed case. 
To conclude, in Section \ref{Sect: conclusion} we summarize how the whole BV construction can be entirely performed in the framework of noncommutative geometry. We again refer to {\small{$\mathrm{U}(n)$}}-gauge theories induced by spectral triples on the matrix algebras {\small{$M_{n}(\mathbb{C})$}}.\\
\\
\noindent
{\em Acknowledgments:} Here we further develop a project whose preliminary results were found in collaboration with W. D. van Suijlekom, whom the author would like to thank for many inspiring conversations. 

\section{The BV construction: origin and physical relevance}
\label{Sect: intro to BV}

\noindent
The Batalin-Vilkovisky (BV) construction was first discovered \cite{BV2, BV3, BV1} as a cohomological solution to the problem of determining a procedure to quantize gauge theories via the path integral approach, where the presence of local symmetries in the action functional does not allow one to straightforwardly apply the standard perturbative approach to perform the quantization.\\
\\
In this section we briefly review the steps that led to the discovery of this construction, addressing the subject using the language of algebraic geometry and mostly focusing on the case of finite dimensional gauge theories. For completeness, we would like to mention that this formalism has been widely investigated also from other purely mathematical perspectives. In particular, the topic has been approached using techniques from {\emph{functional analysis}} (cf. K. Fredenhagen  and K. Rejzner \cite{Rejzner2, Rejzner1,  Rejzner3}) as well as following a more {\emph{algebraic approach}}: O. Gwilliam, first with K. Costello \cite{Costello2, Costello1} and then with R. Haugseng \cite{Gwilliam_H}, described the BV construction using the language of {\small{ $\infty$}}-categories. Recently, these two approaches has been related in \cite{Owen_Kasja} for the case of free field theories. In addition to that, the BV construction has been analyzed from a more {\emph{differential-geometric}} point of view by A. S. Cattaneo, P. Mnev and N. Reshetikhin in a series of papers \cite{pavel_cattaneo_r1, pavel_cattaneo_r3, pavel_cattaneo_r2,  pavel_cattaneo_w}. Within this scenario, this article aims to extend what started in \cite{articolo_BV} together with Van Suijlekom and continued in \cite{articolo_ICMP2022}, that is, adding to the above approaches an analysis of the BV construction from the point of view of noncommutative geometry.\\
\\ 
To start, we recall the notion of {\emph{gauge theory}} as it naturally appears in the context of noncommutative geometry (cf. Section \ref{Sect: intro to NCG and gauge theories}): indeed, while classically a gauge theory is given in terms of connections on a principal bundle, whose structure group encodes the symmetry of the physical theory described, here the notion is, in some aspects, more general.

\begin{definition}
\label{Def: gauge theory}
For a real  vector space $X_{0}$ and a  real-valued functional  $S_{0}$  on $X_{0}$,  let 
$F\colon \mathcal{G} \times X_{0} \rightarrow X_{0}$ be a group action on $X_{0}$ for a given group $\mathcal{G}$.
Then the pair $(X_{0}, S_{0})$ is called a {\em gauge theory with gauge group $\mathcal{G}$} if 
 $$S_{0}(F(g, \varphi)) = S_{0}(\varphi)$$
for all $\varphi \in X_{0}$ and $g \in \mathcal{G}$. The space $X_{0}$ is  referred to as the {\em configuration space}, an element $\varphi \in X_{0}$ is called a {\em gauge field} and $S_{0}$ is the {\em action functional}. 
\end{definition}

To quantize a physical theory {\small{$(X_{0}, S_{0})$}} via the path integral approach, one faces the problem of computing integrals of the following type: 
$$\langle \mathcal{O} \rangle = \frac{1}{Z}\int_{X_{0}} \mathcal{O} e^{\frac{i}{\hslash}S_{0}} [d \mu], \quad \quad \mbox{where} \quad \quad Z:=  \int_{X_{0}} e^{\frac{i}{\hslash}S_{0}} [d \mu].
$$
Hence, given an observable {\small{$\mathcal{O}$}}, that is a functional on {\small{$X_{0}$}}, its expectation value {\small{$\langle \mathcal{O} \rangle$}} is determined by a quotient of integrals where $Z$ is the {\emph{partition function}} of the theory and $d \mu$ denotes a measure on {\small{$X_0$}}. First introduced by Feynman \cite{Feynman} in 1965, over the years the notion of path integral has been developed. However, the problem of rigorously defining a measure $d \mu$ on an (infinite-dimensional) configuration space {\small{$X_0$}} is still mostly unsolved and usually partially sidestepped by applying a {\emph{perturbative approach}}: this involves redefining an ill-defined integral by imposing the stationary phase approximation, which can be proved in the finite-dimensional context and is then used as definition in the infinite-dimensional setting. The explicit computation of a path integral is then performed using the tool provided by the {\emph{Feynman diagrams}} {\small{$\Gamma$}}: 
\begin{multline}
 \label{Feynman_eq}
\int_{X_{0}} e^{\frac{i}{\hslash}S_{0}} [d \mu] \underset{\hslash \rightarrow 0}{\backsim} \\
 \sum_{x_{0} \in \{ \mbox{\footnotesize{crit. pts }}S_{0}\}} e^{\frac{i}{\hslash} S_{0}(x_{0})} \ |\det S_{0}^{\prime \prime}(x_{0})|^{-\frac{1}{2}} \ e^{\frac{\pi i}{4} \sign(S^{\prime\prime}_{0}(x_{0}))} (2 \pi \hslash)^{\frac{\dim X_{0}}{2}} \sum_{\Gamma} \frac{\hslash^{- \chi(\Gamma)}}{|\Aut(\Gamma)|} \Phi_{\Gamma}.
\end{multline}
In this formula, a role is played by the value of the action functional \siS \ and of its Hessian at the critical points of \siS \ as well as by the topological characteristics of the Feynman diagrams, such as their Euler characteristic {\small{$\chi(\Gamma)$}} and the order of their automorphism groups {\small{$\Aut(\Gamma)$}}. Beside the details, the main relevance of the formula \eqref{Feynman_eq} is in allowing one to go from an ill-defined integral to a sum. However, generally in the infinite-dimensional context, to prevent a pivotal divergence, the action {\small{$S_0$}} has to have only {\emph{isolated}} {\emph{non-degenerate}} critical points. It is precisely this crucial condition that is not satisfied by gauge-invariant functionals, where the presence of a gauge symmetry makes the critical points appear in orbits. A first instinct might suggest to solve this problem by replacing the domain of integration with the quotient {\small{$X_{0}/\mathcal{G}$}}, where all physically-equivalent points are identified. However, this approach often leads to a quotient {\small{$X_{0}/\mathcal{G}$}} with a degenerate structure, not suitable to enter the computation as domain of integration. A different approach was then suggested by Faddeev and Popov in 1967 \cite{Faddeev-Popov}: instead of reducing the domain of integration, they proposed to eliminate the divergences appearing for a gauge theory in \eqref{Feynman_eq} by introducing extra (non-physical) fields, suggestively called {\em ghost fields}. 

\begin{definition}
A {\em field/ghost field} $\varphi$ is a graded variable characterized by two integers:
 $$\deg(\varphi) \in \mathbb{Z} \quad \mbox{ and } \quad \epsilon(\varphi) \in \{ 0, 1 \},  \quad \mbox{ with } \quad \deg(\varphi) = \epsilon(\varphi) \quad (\mbox{mod} \ \mathbb{Z}/2).$$
$\deg(\varphi)$ is the {\em ghost degree}, while $\epsilon(\varphi)$ is the {\em parity}, which distinguishes between the bosonic case, where $\epsilon(\varphi)=0$ and $\varphi$ behaves as a real variable, and the fermionic case, where $\epsilon(\varphi)=1$ and $\varphi$ behaves as a Grassmannian variable: 
$$\varphi \psi = - \psi \varphi, \quad \quad \mbox{ and }\quad\quad \varphi^{2} =0, \quad \quad  \mbox { if } \quad \epsilon(\varphi) = \epsilon(\psi) = 1.$$
\end{definition}

\noindent
While Faddeev and Popov just introduced ghost fields of ghost degree $1$, as they were interested in a perturbative approach to Yang-Mills theory, the idea of considering ghost fields of higher degree has to be attributed to Becchi, Rouet, Stora \cite{BRS, BRS3, BRS2} and (independently) Tyutin \cite{T} and goes under the name of the {\emph{BRST construction}}. A further development was suggested simultaneously and independently by Zinn-Justin \cite{Zinn-Justin} and Batalin and Vilkovisky \cite{BV2, BV3, BV1} and is now known as the {\emph{BV formalism}}.

\subsection{The BV extension}
To put the entire BV formalism in a nutshell we could say that it can be viewed as an extension process: given an initial gauge theory {\small{$(X_{0}, S_{0})$}}, it determines a new BV-extended theory {\small{$(\widetilde{X}, \widetilde{S})$}}
$$\begin{array}{ccc}
(X_{0}, S_{0}) & ---------\rightarrow  &(\widetilde{X}, \widetilde{S})\\
\mbox{\small{initial gauge theory}} & \mbox{\tiny{+ ghost/anti-ghost fields}} & \mbox{\small{BV-extended theory}}
\end{array}
$$
where the {\em extended configuration space} {\small{$\widetilde{X}$}} is obtained via the introduction of {\em ghost/an-ti-ghost fields} to the initial configuration space:
$$\widetilde{X} = X_{0} \cup \{ \mbox{ghost/anti-ghost fields} \},$$
and the {\em extended action} {\small{$\widetilde{S}$}} is defined by adding extra terms to the initial action {\small{$S_{0}$}} depending on the ghost/anti-ghost fields:
$$\widetilde{S} = S_{0} + \mbox{terms depending on ghost/anti-ghost fields}.$$
What distinguishes the BV construction from the BRST construction is indeed the presence of {\emph{antifields/anti-ghost fields}}, which are used to give the BV-extended configuration space \swX \ a structure which resembles the one of a shifted cotangent bundle.
 
 \begin{definition}
Given a field/ghost field $\varphi$, the corresponding {\em antifield/anti-ghost field} $\varphi^{*}$ is a graded variable with
$$\deg(\varphi^{*}) = - \deg(\varphi) -1, \quad \quad \mbox{ and } \quad \quad \epsilon(\varphi^{*}) = \epsilon(\varphi) +1, \quad (\mbox{mod} \ \mathbb{Z}/2).$$
\end{definition}
\noindent
{\emph{Note:}} The terminology just introduced differs from the one classically presented in the physics literature. Indeed, it explicitly distinguishes between {\em antifields} and {\emph{anti-ghost fields}}, where this second term is used to identify the antifields corresponding to the ghost fields. The term \textquotedblleft antifields\textquotedblright\ is then devoted to designate only the antifields corresponding to the initial fields in \siX.  \\
\\
In the context we are analyzing, a BV-extended theory has the mathematical structure recalled in the following definition. 

\begin{definition}
\label{def: extended theory}
Given a gauge theory {\small{$(X_{0}, S_{0})$}}, a {\em BV-extended theory} associated to it is a pair {\small{$(\widetilde{X}, \widetilde{S})$}}, where the extended configuration space {\small{$\widetilde{X} = \oplus_{i \in \mathbb{Z}} [\widetilde{X}]^{i}$}} is a {\small{$\mathbb{Z}$}}-graded super vector space decomposable as
\begin{equation}
\label{Eq: decomposition ext conf sp}
\widetilde{X} \cong \mathcal{F} \oplus \mathcal{F}^{*}[1], \quad \quad \mbox{ with } \quad [\widetilde{X}]^{0} = X_{0}
\end{equation}
where {\small{$\mathcal{F} = \oplus_{i \geqslant 0} \mathcal{F}^{i}$}} is a graded locally free {\small{$\mathcal{O}_{X_{0}}$}}-module with homogeneous components of finite rank, for {\small{$\mathcal{O}_{X_{0}}$}} the algebra of regular functions on {\small{$X_{0}$}}. Concerning the extended action {\small{$\widetilde{S} \in [\mathcal{O}_{\widetilde{X}}]^{0}$}}, it is a real-valued regular function on {\small{$\widetilde{X}$}}, with {\small{$\widetilde{S}|_{X_{0}}=S_{0}$}}, {\small{$\widetilde{S}\neq S_{0}$}}, such that it solves the {\em classical master equation}, i.e., 
\begin{equation}
\label{classical master eq}
\{\widetilde{S}, \widetilde{S}\}=0,
\end{equation}
where {\small{$\{ -, -\}$}} denotes the graded Poisson structure on the graded algebra {\small{$\mathcal{O}_{\widetilde{X}}$}}.
 \end{definition}

\noindent
{\emph{Note:}} In the decomposition given in \eqref{Eq: decomposition ext conf sp}, $\mathcal{F}$ describes the fields/ghost-fields content of {\small{$\widetilde{X}$}} while {\small{$\mathcal{F}^{*}[1]$}} determines the antifields/anti-ghost fields part, with {\small{$\mathcal{F}^{*}[1]$}} that denotes the shifted dual module of $\mathcal{F}$:
 $$\mathcal{F}^{*}[1] = \oplus_{i \in \mathbb{Z}} \big[ \mathcal{F}^{*}[1] \big]^{i} \quad \quad \mbox{ with } \quad \big[ \mathcal{F}^{*}[1] \big]^{i} = \big[ \mathcal{F}^{*} \big]^{i+1}.$$
Hence, the extended configuration space \swX\ must have a symmetric structure between fields/ghost fields on one side and antifields/anti-ghost fields on the other. Next, {\small{$\mathcal{O}_{\widetilde{X}}$}} denotes the algebra of real-valued regular functions on {\small{$\widetilde{X}$}}. This algebra has a graded structure naturally induced by the grading on \swX\ and it can be endowed with a graded Poisson structure induced by bracket 
$$\{ -, -\}: [\mathcal{O}_{\widetilde{X}}]^{\bullet} \rightarrow [\mathcal{O}_{\widetilde{X}}]^{\bullet +1}$$ 
of degree {\small{$1$}}. This structure is completely determined by requiring that, on the generators of {\small{$\widetilde{X}$}}, it satisfies the following conditions
\begin{equation}
\label{Eq: pairing fields/antifields}
\big\{ \beta_{i}, \beta_{j}\big\}= 0 , \quad \quad \quad \big\{ \beta^{*}_{i}, \beta_{j}\big\} = \delta_{ij} \quad  \quad \mbox{ and } \quad \quad \big\{ \beta^{*}_{i}, \beta^{*}_{j}\big\}=0 
\end{equation}
for {\small{$\beta_{i} \in \mathcal{F}^{p}$}} and {\small{$\beta^{*}_{i} \in \big[ \mathcal{F}^{*}[1] \big]^{-p-1}$}}, where the only non-zero contribution comes from the pairing of a field/ghost field with the corresponding antifield/anti-ghost field. Finally, even though it is not explicitly required in the notion of BV-extended theory, in what follows we will consider pairs {\small{$(\widetilde{X}, \widetilde{S})$}} with {\em finite level of reducibility}. 

\begin{definition}
\label{def: level of reducibility}
A BV-extended theory {\small{$(\widetilde{X}, \widetilde{S})$}} is said to be 
{\em reducible with level of reducibility} {\small{$L$}} if the {\small{$\mathbb{Z}_{\geqslant 0}$}}-graded and finitely-generated {\small{$\mathcal{O}_{X_{0}}$}}-module {\small{$\mathcal{F}$}} entering the decomposition of \swX\ as {\small{$\widetilde{X}= \mathcal{F} \oplus \mathcal{F}^{*}[1]$}} has homogeneous components of at most degree {\small{$L+1$}}. That is, if 
$$\widetilde{X}= \mathcal{F} \oplus \mathcal{F}^{*}[1] \quad \mbox{with} \quad \mathcal{F} = \oplus_{k=0}^{L+1}\mathcal{F}_{k}.$$
In case {\small{$L=0$}}, the theory is called {\em irreducible}.
\end{definition}

At this point, the next natural question to ask is {\emph{how}}, given an initial gauge theory, one can determine the associated BV-extended theory. In particular, the first problem to solve is to establish the type and number of ghost/anti-ghost fields that have to be introduced in \swX. Answering this question goes beyond the purposes of this brief recall of the BV formalism. For a detailed description of the construction in the infinite-dimensional setting we refer to \cite{Mnev_book}. For the finite-dimensional case, that is, the case when the underlying configuration space is given by an affine variety, the standard construction relies on the computation of the Koszul-Tate resolution \cite{Tate} of the Jacobian ideal {\small{$J(S_{0})$}}, with 
$$J(S_{0}):= \langle \partial_{1}S_{0}, \dots, \partial_{n}S_{0} \rangle, \quad \mbox{for } n=\dim X_{0}$$
over the ring {\small{$\mathcal{O}_{X_{0}}$}}. A method is presented by Felder and Kazhdan in \cite{felder}, which uses the entire Koszul-Tate resolution, ensuring the gauge-invariance of the resulting BV-extended theory, while in \cite{primo_articolo} we explain how to select a finite number of ghost fields, which reflects the complexity of the gauge symmetry of the theory. Hence, having a finite family of additional fields in \swX\ allows to determine an {\emph{exact}} solution of the classical master equation as extended functional {\small{$\widetilde{S}$}}. 

\subsection{The BV cohomology complex}
\label{Subsect: The BV cohomology complex}
Having recalled the notion of a BV-extended theory allows to state the one of {\emph{BV cohomology complex}}. Indeed, the condition required in Definition \ref{def: extended theory} that the extended action \swS\ has to solve the classical master equation is equivalent to imposing that the operator {\small{$d_{\widetilde{S}}:= \{ \widetilde{S}, - \}$}} defines a coboundary operator for the \cBRST  \ complex. In other words, the classical master equation condition is what ensures that any BV-extended theory naturally induces a so-called BV cohomology complex.

\begin{definition}
\label{definition BV complex}
Given a BV-extended theory {\small{$(\widetilde{X}, \widetilde{S})$}}, the induced {\em BV cohomology complex} is a cohomology complex whose cochain spaces {\small{$\mathcal{C}_{\BVt}^{i}(\widetilde{X}, d_{\widetilde{S}})$}} and coboundary operator {\small{$d_{\widetilde{S}}$}} are defined as follows, respectively: 
$$\mathcal{C}_{\BVt}^{i}(\widetilde{X}, d_{\widetilde{S}}) := [\text{Sym}_{\mathcal{O}_{X_{0}}}(\widetilde{X})]^{i} = [\mathcal{O}_{\widetilde{X}}]^{i},  $$
for {\small{$ i \in \mathbb{Z}$}}, with {\small{$\text{Sym}_{\mathcal{O}_{X_{0}}}(\widetilde{X})$}} the {\small{$\mathbb{Z}$}}-graded symmetric algebra generated by {\small{$\widetilde{X}$}} on the ring {\small{$\mathcal{O}_{X_{0}}$}}, and 
$$d_{\widetilde{S}}: \mathcal{C}_{\BVt}^{\bullet}(\widetilde{X}, d_{\widetilde{S}}) \rightarrow \mathcal{C}_{\BVt}^{\bullet +1}(\widetilde{X}, d_{\widetilde{S}}), \quad \quad \mbox{ with } \quad \quad d_{\widetilde{S}}:= \{ \widetilde{S}, - \}$$ 
for {\small{$\{ - , - \}$}} denoting the Poisson bracket structure on {\small{$\mathcal{O}_{\widetilde{X}}$}}.
\end{definition}

Together with the BRST cohomology complex, whose definition we will shortly recall, this complex is the main object of investigation of this paper. One of the reasons this complex is interesting, is that its cohomology groups capture physical properties of the initial gauge theory (see for example in \cite{Rejzner3} and \cite{Mnev_book}). In particular, one can argue that the $0$-degree BV cohomology group {\small{$\mathcal{H}^{0}_{\BVt}(\widetilde{X}, d_{\widetilde{S}})$}} describes the classical observables, that is, the on-shell gauge-invariant functionals, of the initial gauge theory \siXS\ while the higher order groups {\small{$\mathcal{H}_{\BVt}^{\bullet}(\widetilde{X}, d_{\widetilde{S}})$}} describe anomalies and obstructions to the quantization of the theory.
 
\subsection{Auxiliary fields and the gauge-fixing process}
\label{Subsec: Auxiliary fields and the gauge-fixing process}
Despite of all the information summarized in it, computing the BV cohomology complex might not be the ultimate goal of our analysis. In particular, one could be interested in computing amplitudes and S-matrix elements for the theory. If this is the case, at this point the computation cannot be performed straightforwardly, due to the presence of antifields/anti-ghost fields in the configuration space \swX. Hence, one could be interested in removing all the antifields/anti-ghost fields, without canceling the positive effect, ascribable to the presence of the ghost fields, of eliminating the redundancy of the gauge symmetry in the action. This goal is reached by performing a {\emph{gauge-fixing process}}: given a BV-extended theory, we want to determine a new pair 
{\small{$(\widetilde{X}, \widetilde{S})|_{\Psi}$}} where neither the {\em gauge-fixed configuration space} {\small{$\widetilde{X}|_{\Psi}$}} nor the {\em gauge-fixed action} {\small{$\widetilde{S}|_{\Psi}$}} depends on antifields/anti-ghost fields. This goal is reached by defining:
$$\widetilde{X}|_{\Psi}:= \big[\mathcal{F} \oplus\mathcal{F}^{*}[1]\big]|_{\varphi_{i}^{*} = \frac{\partial \Psi}{\partial \varphi_{i}}} \quad \mbox{and} \quad \widetilde{S}(\varphi_{i}, \varphi^{*}_{i})|_{\Psi} := \widetilde{S}\Big(\varphi_{i}, \varphi_{i}^{*} =\frac{\partial \Psi}{\partial \varphi_{i}}\Big),$$
for $\Psi$ a so-called {\emph{gauge-fixing fermion}}.

\begin{definition}
\label{def gauge fixing fermion}
Given an extended configuration space {\small{$\widetilde{X}$}}, a {\em gauge-fixing fermion} $\Psi$ on it is a regular function {\small{$ \Psi \in \mathcal{O}^{-1}_{\mathcal{F}}$}}. That is, a regular function depending only on fields/ghost fields, of total degree $-1$ and hence odd parity. 
\end{definition}

Thus the new gauge-fixed configuration space {\small{$\widetilde{X}|_{\Psi}$}} is defined to be the Lagrangian submanifold determined by imposing the {\em gauge-fixing conditions} {\small{$\varphi_{i}^{*} =\partial \Psi/\partial \varphi_{i}$}}. That is, by replacing every antifield/anti-ghost field {\small{$\varphi^{*}_{i} \in \mathcal{F}^{*}[1] $}} with the partial derivative of {\small{$\Psi$}} with respect to the corresponding field/ghost field {\small{$\varphi_{i}$}} while {\small{$\widetilde{S}|_{\Psi}$}} is the restriction of the action \swS \ to {\small{$\widetilde{X}|_{\Psi}$}}.\\
\\
However, an additional obstacle occurs in the application of the gauge-fixing procedure if the BV-extended theory \swXS \ we are considering has been obtained by applying the classical BV construction described above: indeed, the absence of ghost fields of negative ghost degree does not allow to define a gauge-fixing fermion straightforwardly. Hence, an intermediate step needs to be taken: after having performed the first extension by adding ghost/anti-ghost fields, we have to further enlarge the configuration space \swX \ via the introduction of auxiliary fields and to add extra terms at the extended action \swS, obtaining a so-called {\emph{total theory}}: 

$$\begin{array}{ccccc}
(X_{0}, S_{0}) & -----\rightarrow  &(\widetilde{X}, \widetilde{S})& -----\rightarrow  &(X_{\t}, S_{\t})\\
\mbox{\small{initial theory}} & \mbox{\tiny{+ ghosts/anti-ghosts}} & \mbox{\small{BV-extended theory}} & \mbox{\tiny{+ auxiliary fields}} & \mbox{\small{total theory}}
\end{array}
$$

In order not to change the induced BV complex, the  auxiliary fields have to be trivial from a cohomological point of view. That is, the BV complex induced by the new pair \stXS\ should be quasi-isomorphic to the one induced by the BV-extended theory \swXS. Consequently, we ask that the auxiliary fields determine {\emph{contractible pairs}} in cohomology (cf. \cite{BV2, BV1}).

\begin{definition}
\label{auxiliary pair}
An {\em auxiliary pair} is a pair of fields $(B, h)$ such that their ghost degrees and parities satisfy the following relations:
\begin{equation}
\begin{array}{lr}
\deg(h)= \deg(B) +1; \quad \quad & \quad \quad \epsilon(h)= \epsilon(B) + 1 \mbox{ (mod 2) }. 
\label{ghost degree and parity fields in trivial pair}
\end{array}
\end{equation}
\end{definition}

\begin{definition}
Given a BV-extended theory {\small{$(\widetilde{X}, \widetilde{S})$}} and an auxiliary pair {\small{$(B, h)$}}, the corresponding {\em total theory} {\small{$(X_{\t}, S_{\t})$}} has a {\em total configuration space} {\small{$X_{\t}$}} defined as the {\small{$\mathbb{Z}$}}-graded super vector space generated by {\small{$\widetilde{X}$}}, {\small{$(B, h)$}} and their corresponding antifields {\small{$(B^{*}, h^{*})$}}
$$X_{\t}:= \langle \widetilde{X}, B, h, B^*, h^* \rangle, $$
and a {\em total action} {\small{$S_{\t}$}} given by the following sum:
$$S_{\t}:= \widetilde{S} + S_{\text{aux}},\quad \quad \mbox{ where } \quad \quad S_{\text{aux}}:= h B^{*}.$$
\end{definition}

\noindent
{\emph{Note}}. By construction, the theories {\small{$(\widetilde{X}, \widetilde{S})$}} and {\small{$(X_{\t}, S_{\t})$}} satisfy similar properties, the only difference lying in the fact that {\small{$X_{\t}$}} may also contain negatively graded fields. In particular, {\small{$X_{\t}$}} also exhibits the symmetry between fields/ghost fields and antifields/anti-ghost fields content, being decomposable as
$$X_{\t} = \mathcal{E} \oplus \mathcal{E}^{*}[1],$$
where {\small{$\mathcal{E}$}} is a $\mathbb{Z}$-graded finitely generated {\small{$\mathcal{O}_{X_{0}}$}}-module. Moreover, the graded Poisson structure on {\small{$\widetilde{X}$}} can be extended to one on  {\small{$X_{\t}$}} by imposing that  
$$\{ B^{*}, B\} = \{ h^{*}, h\} =1 $$
while the value of the bracket on any other possible combination of auxiliary fields, corresponding antifields or generators in \swX\ is declared to be zero. As a consequence the total action {\small{$S_{\t}$}} solves the classical master equation on {\small{$\mathcal{O}_{X_{\t}}$}}:
$$\{ S_{\t}, S_{\t}\} =0.$$

In order to properly implement a gauge-fixing procedure further conditions have to be enforced on the gauge-fixing fermion $\Psi$. Indeed, we need that the physically relevant quantities that are computed as 
\begin{equation}
\label{eq: integrale per gauge-fix}
\int_{X_{\t}|_{\psi}} \mathcal{O} \ d(\mathrm{Vol}_{X_{\t}|_{\Psi}}),
\end{equation}
for {\small{$\mathrm{Vol}_{X_{\t}|_{\Psi}}$}} a volume form on {\small{$X_{\t}|_{\Psi}$}} and $\mathcal{O}$ a regular function on {\small{$X_{\t}|_{\Psi}$}}, do not depend on the choice of $\Psi$. This condition is satisfied by imposing that the gauge-fixed action is a {\emph{proper solution}} to the classical master equation (cf. \cite{Fior, GPS, Schw}). This condition is what dictates the number and type of auxiliary pairs that have to be introduced, which results to be determined by the level of reducibility of the BV-extended theory. In addition to the reference quoted above, some details on the gauge-fixing procedure for finite-dimensional gauge theory can also be found in \cite{articolo_cohomology}.

\subsection{The BRST cohomology complex}
\label{Subsection: The BRST cohomology complex}
Once the gauge-fixing process has been implemented, a natural question arises: how does this process effect the BV cohomology complex? Is there a residual BV symmetry on {\small{$(X_{\t}, S_{\t})|_{\Psi}$}} that induces a cohomology complex? The answer to this question is always positive, if we consider the theory {\emph{on-shell}}, that is, if we suppose that the equations of motion {\small{$ \partial S_{\t}|_{\Psi}/\partial \varphi_{i} = 0$}} are satisfied for all {\small{$\varphi_{i} \in \mathcal{E}$}} (cf. \cite{AKSZ}). However, an explicit computation shows that, if we restrict the operator {\small{$d_{S_{\t}}$}} to the algebra {\small{$\mathcal{O}_{X_{\t}|_{\Psi}}$}} of regular functions on the gauge-fixed configuration space, this might still satisfy the coboundary condition also off-shell, depending on the explicit form of the action {\small{$\widetilde{S}$}}. Hence, it is meaningful and physically relevant, to consider what we call the {\emph{BRST cohomology complex}}.  

\begin{definition}
\label{definition of BRST cohomology}
 Given a total theory {\small{$(X_{\t}, S_{\t})$}}, with {\small{$X_{\t} = \mathcal{E} \oplus \mathcal{E}^{*}[1]$}} and {\small{$S_{\t}~\in~[\mathcal{O}_{X_{\t}}]^{0}$}}, together with a gauge-fixing fermion {\small{$\Psi \in [\mathcal{O}_{\mathcal{E}}]^{-1}$}}, the induced {\em BRST complex} {\small{$\mathcal{C}_{\BRST}^{j}(X_{\t}|_{\Psi}, d_{S_{\t}}|_{\Psi})$}} is a cohomology complex with
$$\mathcal{C}_{\BRST}^{j}(X_{\t}, d_{S_{\t}})_{\Psi} := \big[ \Sym_{\mathcal{O}_{X_{0}}}(X_{\t}|_{\Psi})\big]^{j}, \quad \quad \mbox{ and } \quad \quad d_{S_{\t}}|_{\Psi} := \left.\big\{ S_{\t}, - \big\}\right|_{X_{\t}|_{\Psi}} $$
where {\small{$j \in \mathbb{Z}$}} and {\small{$X_{\t}|_{\Psi} \subset X_{\t}$}} is the Lagrangian submanifold defined by the gauge-fixing conditions {\small{$\{ \varphi_{i}^{*} = \partial \Psi/\partial \varphi_{i}\}.$}}
\end{definition}
\noindent
The appearance of this second cohomology complex underlines the richness of this mathematical setting. Even more, also the cohomology groups of this BRST complex turn out to have a physical relevance as they encode interesting physical information about the initial gauge theory \siXS, such as, once again, the space of classical observables:
$$H_{\BRST}^{0}(X_{\t}, d_{S_{\t}})_{\Psi} = \{ \mbox{Classical observables for the initial theory } (X_{0}, S_{0}) \}.$$

\subsection{Toward the BV quantization}
To summarize, up to now we described a procedure that, given an initial gauge theory \siXS, associates to it two cohomology complexes, namely the BV and the BRST complexes. Schematically, under the name of {\emph{classical BV construction}} we collect the following passages:

\begin{picture}(400, 100)
\put(-10, 77){$(X_{0}, S_{0})$}
\put(-2, 65){\footnotesize{initial}}
\put(-12, 58){\footnotesize{gauge theory}}

\put(28, 82){\tiny{$+$ ghost/anti-ghost}}
\put(28, 79){\vector(1,0){65}}

\put(97, 77){$(\widetilde{X}, \widetilde{S})$}
\put(83, 65){\footnotesize{BV-extended}}
\put(97, 58){\footnotesize{theory}}

\put(110, 55){\vector(0,-1){20}}
\put(80, 22){$\mathcal{C}_{\BVt}^{\bullet}(\widetilde{X}, d_{\widetilde{S}})$}
\put(81, 10){\footnotesize{BV complex}}

\put(129, 79){\vector(1,0){57}}
\put(131, 82){\tiny{$+$ auxiliary flds}}

\put(189, 77){$(X_{\t}, S_{\t})$}
\put(198, 65){\footnotesize{total}}
\put(196, 58){\footnotesize{theory}}

\put(153, 23){$\cong$}

\put(207, 55){\vector(0,-1){20}}
\put(182, 22){$\mathcal{C}_{\BVt}^{\bullet}(X_{\t}, d_{S_{\t}})$}
\put(183, 10){\footnotesize{BV complex}}

\put(227, 79){\vector(1,0){48}}
\put(229, 82){\tiny{gauge-fixing}}

\put(278, 77){$(X_{\t}, S_{\t})|_{\Psi}$}
\put(282, 65){\footnotesize{gauge-fixed}}
\put(292, 58){\footnotesize{theory}}

\put(302, 55){\vector(0,-1){20}}
\put(260, 22){$\mathcal{C}_{\BRST}^{\bullet}(X_{\t}, d_{S_{\t}})_{\Psi}$}
\put(272, 10){\footnotesize{BRST complex}}
\end{picture}

\noindent
As seen, both the BV as well as the BRST complex naturally appear in the context of BV-extended gauge theories. But how can all of this help us dealing with orbits of critical points for $S_{0}$? The key result in the BV formalism is the so-called BV theorem, first presented by Batalin and Vilkovisky \cite{BV1, BV2} and then formalized in the context of odd-symplectic manifolds by Schwarz \cite{Schw}. In this theorem, it is proven that the following two integrals coincide
\begin{equation}
\label{eq: BV theorem}
 \int_{X_{0}} e^{\frac{i}{\hslash}S_{0}} [d \mu] = \int_{\mathcal{L}\subseteq X_{t}} e^{\frac{i}{\hslash} S_{q}|_{\mathcal{L}}} [d \mu]    
\end{equation}
for {\small{$\mathcal{L}$}} a Lagrangian submanifold homotopically equivalent to \siX\ in the {\small{$\mathbb{Z}$}}-graded supermanifold {\small{$X_{t}$}} and {\small{$S_{q}$}} an element in the same quantum BV-cohomology class of \siS. In other words, the path integral of a gauge theory can be seen as an integral depending on two inputs: the configuration space \siX, which enters as domain of integration, and the action functional \siS, which appearing in the integrand. As proved in this theorem, the path integral depends on these two parameters quite weakly as the results depends only on the {\emph{homotopy}} class of \siX \ and on the {\emph{quantum BV-cohomology}} class of \siS. Hence, while the integral appearing on the left-hand side in \eqref{eq: BV theorem} is often not well defined, we can assign to it the value obtained by computing the integral on the right-hand side for any choice of {\small{$\mathcal{L} \in [X_{0}]_{\text{hom}}$}} and of {\small{$S_{q} \in [S_{0}]_{\text{\tiny{BV},q}}$}}. In particular, one is interested in finding a pair {\small{$(\mathcal{L}, S_{q})$}} such that {\small{$S_{q}|_{\mathcal{L}}$}} only has isolated singular points, which would then allow one to compute this integral by applying a classical stationary phase approximation for $\hslash \rightarrow 0$.\\
\\
To conclude this short review of the BV construction we only have to briefly recall the notion of {\emph{quantum BV cohomology complex}}, which relies on the {\emph{quantum master equation}} as the classical notion of BV cohomology relates to the classical master equation. To state the quantum master equation one has to first introduce the notion of {\emph{BV Laplacian}}. This definition turns out to be quite technical if one wants to state it in its full generality, that is, in the context of odd-symplectic manifolds. However, as we will consider the BV Laplacian only in the context of affine gauge theories, that is, in a context where the underline configuration space can be globally covered with one special Darboux chart, which in the general context was just a local description, in the setting we will be interested in could be taken as a global definition. For a more detailed presentation of the topic in the context of odd-symplectic manifolds in general, we refer the interested reader to \cite{Mnev_book}.

\begin{definition}
Given \swX \ a {\small{$\mathbb{Z}$}}-graded supervector space, the {\emph{Batalin-Vilkovisky Laplacian}} on \swX \ is the odd second-order operator {\small{$\Delta: \mathcal{O}_{\widetilde{X}} \rightarrow  \mathcal{O}_{\widetilde{X}}$}} with global coordinate description 
$$\Delta: = \sum_{i} \frac{\partial }{\partial \varphi_{i}} \frac{\partial }{\partial \varphi^{*}_{i}},$$
where {\small{$\varphi^{*}_{i}$}} denotes the antifield/anti-ghost field associated to the field/ghost field {\small{$\varphi_{i}$}}.
\end{definition}

\noindent
{\emph{Note:}} using the description of {\small{$\Delta$}} in local coordinates, it is possible to check that this operator satisfies {\small{$\Delta^{2} =0$}}.
 
\begin{definition}
Let {\small{$\hslash$}} denote a formal parameter. Then an element {\small{$S_{q} \in [\mathcal{O}_{\widetilde{X}}]^{0}[[-i\hslash]]$}}, that is, a formal power series in {\small{$-i\hslash$}} with coefficients in the {\small{$0$}}-degree part of the graded commutative algebra {\small{$\mathcal{O}_{\widetilde{X}}$}}, is said to satisfy the {\emph{quantum master equation}} if the following holds:
$$\frac{1}{2} \{ S_{q}, S_{q} \} - i \hslash \Delta(S_{q}) = 0.$$
\end{definition}

\begin{definition}
Let {\small{$S_{q}$}} be a solution to the quantum master equation, with {\small{$S_{q} \in [\mathcal{O}_{\widetilde{X}}]^{0}[[-i\hslash]]$}}. Then the {\emph{quantum BV cohomology complex}} is defined to be a cohomology complex with space of cochains and coboundary operator given, respectively, by:
$$\mathcal{C}_{\text{\tiny{BV,q}}}^{\bullet}(\mathcal{O}_{\widetilde{X}}[[-i\hslash]], \delta_{q}): = [\mathcal{O}_{\widetilde{X}}]^{\bullet}[[-i\hslash]], \quad \delta_{q}:= \frac{1}{2} \{ S_{q}, - \} - i \hslash \Delta( - ). $$
\end{definition}

\noindent
{\emph{Note}}: the fact that {\small{$\delta_{q}$}} as defined above determines a coboundary operator for the quantum BV cohomology complex is a direct consequence of {\small{$S_{q}$}} being a solution to the quantum master equation.\\
\\
\noindent
The aim of this article is to present in details the BV construction, from the introduction of ghost fields to the resolution of the quantum master equation, for gauge theories induced by finite spectral triples and to prove how this construction can be described entirely in the language of noncommutative geometry. 

\section{From spectral triples to gauge theories: notions and constructions}
\label{Sect: intro to NCG and gauge theories}
\noindent
In the quest for a mathematical formalism where to describe all the four fundamental interactions appearing in nature, geometry has proved several times to provide an extremely successful setting. Among the proposed approaches, string theory prescribes a micro-structure of the 4-dimensional spacetime where, over each point, one can find extra 6-dimensions rolled up in the shape of a Calabi-Yau manifold. Contrary, noncommutative geometry \cite{Connes_Intro_NCG, Connes_libro_NCG}  does not add extra dimensions but suggests to make the space slightly noncommutative: specifically, one considers models where, over each point of the underlying 4-dimensional spacetime, there is a finite-dimensional noncommutative space \cite{CC96,Connes_Stand_Model,Connes2}. These noncommutative spaces are called {\emph{spectral triples}} and, interestingly, the {\emph{finite}} spectral triples in particular account for the particle content of the physical models. On the other hand, when dropping the finite-dimensionality condition, a spectral triple can be viewed as an extension of the classical concept of compact Riemannian spin manifold \cite{Connes_Reconstruction_Theorem}. \\
\\Next to its relevance from a purely mathematical point of view, very remarkably this new geometrical object has also a deep relation to quantum field theory. Indeed, given a spectral triple, finite or infinite dimensional, commutative or not, it always encodes a {\emph{gauge theory}}, that is, a physical theory invariant under the action of a (local) group of symmetries. This is a classical result in the field, which we recall in Proposition \ref{Prop: gauge theory from spectral triple}. Although stated only in the finite-dimensional context, this property also holds for eventually infinite-dimensional spectral triples.\\
\\
This deep relation between spectral triples and gauge theories is disclosed by several results, among which one can find the description of the full Standard Model of particles as product of a canonical and a finite spectral triple: while the canonical (infinite-dimensional) spectral triple represents the spacetime and is responsible for the gravitational part of the theory, it is the finite spectral triple that accounts for the quantum interactions, whose gauge group is {\small{$\mathrm{SU}(3) \times \mathrm{SU}(2)\times \mathrm{U}(1)$}} \cite{CCM07}. Recently, many results have been obtained in the direction of going beyond the Standard Model from a completely mathematical perspective, using the notion of spectral triple as leading concept \cite{Pati_Salam1, Pati_Salam2, Grand_symm}.  All these results further confirm the importance of the role played by finite spectral triples in the context of a mathematical formalization of quantum field theory and the relevance of describing the BV construction for gauge theories induced by this class of triples.\\
\\ 
In this section, we briefly recall the relevant notions regarding spectral triples and their structure, which will play a key role in our construction. 
\begin{definition}
  \label{def spectral triple}
Let {\small{$\mathcal{A}$}} be an involutive unital algebra, {\small{$\mathcal{H}$}} a separable Hilbert space and {\small{$D: \mathcal{H} \rightarrow \mathcal{H}$}} a self-adjoint operator on {\small{$\mathcal{H}$}} with a dense domain. Then the triple {\small{$(\mathcal{A}, \mathcal{H}, D)$}} is a {\em spectral triple} if {\small{$\mathcal{A}$}} can be faithfully represented as bounded operators on {\small{$\mathcal{H}$}} and the operator {\small{$D$}} has a compact resolvent and bounded commutators $[D, a]$ for each $a \in$ {\small{$\mathcal{A}$}}. Moreover, a spectral triple {\small{$(\mathcal{A}, \mathcal{H}, D)$}} is called {\em finite} if the Hilbert space {\small{$\mathcal{H}$}} and hence the algebra {\small{$\mathcal{A}$}} are finite dimensional.
  \end{definition}

\noindent
{\emph{Note:}} By Wedderburn theorem (cf. for example \cite{walter}) one deduces that, in a finite spectral triple {\small{$(\mathcal{A}, \mathcal{H}, D)$}}, the algebra {\small{$\mathcal{A}$}} is necessarily a direct sum of matrix algebras, i.e.
$$ \mathcal{A}\simeq \bigoplus_{i=1}^{k}M_{n_{i}}(\mathbb{C})
 $$
 for positive integers $n_{1}, \dots, n_{k}$. Moreover, again in the finite-dimensional setting, any self-adjoint operator {\small{$D$}} on {\small{$\mathcal{H}$}} would automatically satisfy all additional conditions  on its commutators and its resolvent entering the definition of a spectral triple.\\
 \\
Given a spectral triple {\small{$(\mathcal{A}, \mathcal{H}, D)$}}, its structure could be further enriched by the introduction of a {\emph{real structure}} $J$, first suggested in \cite{Real_structure}, and a {\emph{grading map}} $\gamma$ satisfying certain conditions. Because in what follows the notion of real structure as well as the one of {\emph{KO-dimension}} will play a relevant and physically motivated role, we briefly recall their definitions.

  \begin{definition}
An {\em even spectral triple} {\small{$(\mathcal{A}, \mathcal{H}, D)$}} is one in which the Hilbert space {\small{$\mathcal{H}$}} is endowed with a $\mathbb{Z}/2$-grading $\gamma$, given by a linear map $\gamma: \mathcal{H} \to \mathcal{H}$, such that 
$$   D\gamma = -\gamma D  \quad \text{and} \quad   \gamma a = a \gamma $$
for all $a \in \mathcal{A}$.
  \end{definition}
  
\begin{definition}
\label{def real structure}
A {\em real structure of KO-dimension $n$ (mod 8)} on a spectral triple $(\mathcal{A}, \mathcal{H}, D)$ is  an anti-linear isometry $ J: \mathcal{H}\rightarrow \mathcal{H}$ that satisfies 
$$
J^2 = \epsilon  \quad \text{and} \quad  J D = \epsilon^{\prime} DJ
$$
together with the condition $$ J \gamma = \epsilon^{\prime \prime} \gamma J $$ in the even case. 
The constants $\epsilon$, $\epsilon^{\prime}$ and $\epsilon^{\prime \prime}$  depend on the KO-dimension $n$ (mod 8) as follows:
$$
\begin{array}{|c|cccccccc|}
 \hline
 {n} & \phantom{-}0 & \phantom{-}1 & \phantom{-}2 & \phantom{-}3 & \phantom{-}4 & \phantom{-}5 & \phantom{-}6 & \phantom{-}7\\
 \hline
 \epsilon & \phantom{-}1 & \phantom{-}1 & -1 & -1 & -1 & -1 & \phantom{-}1 & \phantom{-}1\\
 \epsilon^{\prime} & \phantom{-}1 & -1 & \phantom{-} 1 & \phantom{-}1 & \phantom{-}1 & -1 & \phantom{-}1 & \phantom{-}1 \\ 
 \epsilon^{\prime \prime} & \phantom{-}1 & & -1 & & \phantom{-}1 & & -1 &\\
 \hline
\end{array}
$$ \vspace{2mm}\\
Moreover, we require for all $a,b\in \mathcal{A}$ that:
\begin{itemize}
\item[-] the action of $\mathcal{A}$ satisfies the {\em commutation rule}: $\big[a,  J b^{*}J^{-1}\big]=0$;
\item[-] the operator $D$ fulfills the {\em first-order condition}: $[ [ D, a ], J b^{*} J^{-1}] = 0$.
\end{itemize} 
When a spectral triple $(\mathcal{A}, \mathcal{H}, D)$ is endowed with such a real structure $J$, it is said to be a {\em real spectral triple} and denoted by $(\mathcal{A}, \mathcal{H}, D, J)$.
 \end{definition}
 
While the standard definition of a real spectral triple applies both to the finite- and the infinite-dimensional case, when we restrict to consider the finite setting this notion can be relaxed, allowing the operator $D$ not to fully commute or anti-commute with the isometry $J$. This gives what we call a {\emph{mixed KO-dimension}} situation. 

\begin{definition}
\label{Def: mixed KO-dim}
For $(\mathcal{A}, \mathcal{H}, D)$ a finite spectral triple and $J$ an anti-linear isometry on $\mathcal{H}$, we say that $(\mathcal{A}, \mathcal{H}, D, J)$ defines a {\em real spectral triple with mixed KO-dimension} if $J$ satisfies  
$$J^{2}= \pm \mathrm{Id} \qquad \mbox{ and } \qquad [a, Jb^{*}J^{-1}] =0$$
for $ a, b \in \mathcal{A}$, the operator $D$ can be seen as a sum $D=D_1 + D_2$ of two non-zero self-adjoint operators $D_{1}$, $D_{2}$, which {\em anti}-commutes and commutes, respectively, with $J$:
$$JD_1 = - D_1 J  \quad \mbox{ and }  \quad JD_2 = D_2J, $$
and, finally, the first-order condition holds:
$$[[D, a], Jb^{*}J^{-1}] = 0, \qquad ( a, b \in \mathcal{A}).$$
\end{definition}
\noindent
In order to be able to associate to any given spectral triple a naturally induced gauge theory, we still have to introduce a notion of action in the context of noncommutative geometry. There are two definitions of action functionals associated to a spectral triple: while the {\emph{spectral action}}, introduced in \cite{CC97}, is the only natural additive spectral invariant of noncommutative geometry, the {\emph{fermionic action}} \cite{CC96,CCM07} is defined on a subspace of the Hilbert space and can depend on the real structure.
  
\begin{definition}
\label{def spectral action}
Given a spectral triple {\small{$(\mathcal{A}, \mathcal{H}, D)$}} and a positive, even real-valued function $f$, the associated {\em spectral action} {\small{$S_{0}$}} is a functional {\small{$S_{0}: [\Omega^{1}_{D}(\mathcal{A})]_{\text{s.a.}} \rightarrow \mathbb{R}$}} defined to be 
$$S_{0} [\varphi]: = \tr  f\left((D + \varphi)/\Lambda\right),
$$
for $\Lambda$ a real cutoff parameter and $f$ a function with a sufficiently rapid decay at infinity to make $f\left((D + \varphi)/\Lambda\right)$ a traceclass operator. 
\end{definition}

\noindent
{\emph{Notation}:} by {\small{$[\Omega^{1}_{D}(\mathcal{A})]_{\text{s.a.}}$}} we denote the space of inner fluctuations of the operator $D$:
$$[\Omega^{1}_{D}(\mathcal{A})]_{\text{s.a.}}:= \Big\{ \varphi = \sum_{j} a_{j} \big[ D, b_{j}\big]:  \varphi^{*} = \varphi, \ a_{j}, b_{j} \in \mathcal{A}  \Big\}, $$
where we only consider finite sums and $*$ denotes the involution structure on {\small{$\mathcal{A}$}}.\\
\\
\noindent
{\emph{Note:}} In the finite-dimensional setting, a family of suitable functions $f$ is given by the polynomials in {\small{$\mathbb{R}[x]$}}. Moreover, in this context, the trace operator coincides with the standard trace of matrices.

\begin{definition}
\label{def: fermionic action}
For a spectral triple {\small{$(\mathcal{A}, \mathcal{H}, D)$}}  (real spectral triple {\small{$(\mathcal{A}, \mathcal{H}, D, J)$}}) the {\em fermionic action}  on a subspace {\small{$\mathcal{H}^{\prime} \subseteq \mathcal{H}$}} is given by
$$S_{\ferm}[\varphi] 
:= \frac{1}{2} \langle \varphi , D \varphi \rangle 
\quad \Big(S_{\ferm}[\varphi] : = \frac{1}{2} \langle J\varphi , D \varphi \rangle  \Big)
$$
for $ \varphi \in \mathcal{H}^{\prime}$, where {\small{$\langle \ , \ \rangle$}} denotes the inner product structure on {\small{$\mathcal{H}$}}.
 \end{definition}
\noindent
We conclude by recalling the classical result (cf. \cite{marcolli}) on how to construct, for a finite spectral triple {\small{$(\mathcal{A}, \mathcal{H}, D)$}}, the induced gauge theory {\small{$(X_{0}, S_{0})$}}, with gauge group {\small{$\mathcal{G}$}}.

\begin{prop}
\label{Prop: gauge theory from spectral triple}
Given {\small{$(\mathcal{A}, \mathcal{H}, D)$}} a finite spectral triple, let {\small{$X_{0}$}} denote the space of inner fluctuations for the operator {\small{$D$}}, {\small{$X_{0} := [\Omega^{1}_{D}(\mathcal{A})]^{*}$}}, {\small{$\mathcal{G}$}} be the space of unitary elements of $\mathcal{A}$ 
$$ \mathcal{G} : =\mathcal{U}(\mathcal{A}) = \{ u \in \mathcal{A}: u u^{*}= u^{*}u =1 \},$$
which acts on $X_0$ via the map 
$$
F:\mathcal{G}  \times  X_{0}  \to X_{0},\qquad
 (u, M) \mapsto  u M u^{*} + u[D, u^{*}], 
$$
and, finally, let $S_0$ be the spectral action defined on {\small{$X_{0}$}} as
$$S_{0} [\varphi]: = \tr \big( f(D + \varphi)\big) ,$$ 
for any {\small{$\varphi\in X_0$}} and some {\small{$f\in \mathbb{R}[x]$}}. Then the pair {\small{$(X_{0}, S_{0})$}} is a gauge theory with gauge group {\small{$\mathcal{G}$}}.
\end{prop}

\noindent
\subsection{A $\mathrm{U}(n)$-model.} Before proceeding with the construction, let us introduce an example. This will provide a model on which we could explicitly describe results that might appear quite abstract otherwise. Moreover, this running example will ensure the existence of a BV spectral triple associated to an initial spectral triple as described in the next section as well as of other structures we will introduce in the continuation of this article. For completeness, while this model allows to already cover a large class of physical theories, it is still not the most general class of examples one could consider: indeed, as previously recalled, Wedderburn's theorem would allow the algebra {\small{$\mathcal{A}$}} in a spectral triple to be a finite direct sum of different matrix algebras.\\
\\
Let \isp\ denote the following finite spectral triple: 
$$(\mathcal{A}_{0}, \mathcal{H}_{0}, D_{0}): =  (M_{n}(\mathbb{C}), \mathbb{C}^{n}, D_{0}), $$
for {\small{$n \geqslant 2$}}, where {\small{$D_{0}$}} is a self-adjoint complex {\small{$n \times n$}}-matrix. By applying the classical construction recalled in Proposition \ref{Prop: gauge theory from spectral triple}, the above spectral triple yields an initial gauge theory \siXS\
with initial configuration space \siX, initial action functional \siS\ and gauge group {\small{$\mathcal{G}$}} respectively defined as:
$$X_0=\{ M \in M_n(\mathbb C): M^{*} = M \},
\quad  S_0[M] = \tr f (M)
 \quad  \text{and} \quad  \mathcal G = \mathrm{U}(n),$$
where $*$ denotes the usual adjoint for a matrix, $f$ is a polynomial in {\small{$\mathbb{R}[x]$}} and {\small{$\mathcal G$}} acts on {\small{$X_0$}} via the adjoint action (an explicit computation of \siX \ can be found in \cite{walter}). \\
\\
{\emph{Note:}} In the rest of the article, when speaking about {\emph{$\mathrm{U}(n)$-gauge theories}}, we will always refer to gauge theories of this kind, that is, affine, {\small{$\mathrm{U}(n)-$}}invariant gauge theories.\\
\\
\noindent 
By fixing as a basis for \siX \ the one given by the generalized Gell-Mann matrices {\small{$\left\lbrace\sigma_{a}\right\rbrace_{a=1}^{n^{2}-1}$}} with the identity matrix {\small{$\sigma_{n^2}: = {\mathrm{Id}}$}}, \siX\ appears to be isomorphic to an $n^{2}$-dimensional real affine space, where the coordinates {\small{$x_{1}, \dots, x_{n^{2}}$}} are the dual basis of {\small{$\left\lbrace\sigma_{a}\right\rbrace_{a=1}^{n^{2}}$}}:
$$X_{0} \cong \mathbb{A}^{n^{2}}_{\mathbb{R}}:= \left\lbrace (x_{1}, \dots, x_{n^{2}}) \in \mathbb{R}^{n^{2}}\right\rbrace.$$
\noindent 
As a consequence, the algebra {\small{$\mathcal{O}_{X_{0}}$}} of real-valued regular functions on \siX\ is simply the coordinate ring of the affine space {\small{$\mathbb{A}^{n^{2}}_{\mathbb{R}}$}}, that is, the algebra of polynomials over {\small{$\mathbb{R}$}} in {\small{$n^{2}$}} coordinates:
$$\mathcal{O}_{X_{0}} \cong \mathbb{R}[x_{1}, \dots, x_{n^{2}}].$$
For what concerns the initial action \siS, it is an element in {\small{$\mathcal{O}_{X_{0}}$}} such that it is invariant under the adjoint action of the gauge group {\small{$\mathcal{G} = \mathrm{U}(n)$}}. By applying classical results \cite{Casimir2, Casimir1} in invariant theory for Lie algebras and Lie groups, one can conclude that the action functional \siS\ should be a polynomial in the Casimir operators, whose order goes from {\small{$k=2$}}, with the quadratic Casimir, to {\small{$k=n$}}. In particular, if one supposes the action 
functional \siS \ to be a function only of the quadratic Casimir, it would then have the following form:
 
\begin{equation}
S_{0} = \sum_{k=0}^{r} \mbox{ }(x_{1}^2 + x_{2}^2 + \dots + x_{n^2-1}^2)^k g_{k}(x_{n^2}),
\label{eq: S_0 quadratic Casimir}
\end{equation}
for {\small{$r \in \mathbb{Z}_{\geqslant 0}$}} and {\small{$g_k(x_{n^2})\in \mathbb{R}[x_{n^2}]$}}. As the variable {\small{$x_{n^2}$}} is the dual variable to the identity matrix in the basis, enforcing that the action functional \siS\ is invariant under the action of the gauge group {\small{$\mathrm{U}(n)$}} does not impose any condition on it, leaving {\small{$x_{n^2}$}} to be a free variable in \siS. \\
\\
\noindent
The remainder of this article will be devoted to performing the BV construction on this particular class of gauge theories, which are defined starting from a finite spectral triple. As mentioned in Section \ref{Sect: intro to BV}, the reason for applying the BV formalism is the presence of orbits of critical points for the action functional \siS. Even though an explicit study of the critical locus of \siS \ would require taking into account the function $f$ used to define it, one can still make some preliminary observations. \\
\\
{\emph{The critical locus of \siS.}} Under the mild condition of asking the polynomial {\small{$f$}} not to have terms of order lower or equal to {\small{$1$}}, one can verify that the origin is always a critical point for \siS. However, the origin is also the only orbit described by the action of the gauge group {\small{$\mathrm{U}(n)$}} on \siX \ which is given by just one point, being then isolated, even though often not {\emph{regular}}. Moreover, in the generic case, the origin is not the only critical point of \siS. In particular, one can see this for the class of action functionals \siS \
 which depend only on the quadratic Casimir and that appear in Equation \eqref{eq: S_0 quadratic Casimir}. With a direct computation one can check that the line defined by the equations {\small{$x_{1} = \dots = x_{n^2- 1} =0$}} is always a singular line for \siS, if one assumes {\small{$g_{0}=0$}}, that is, if {\small{$f$}} does not have terms of order lower or equal to {\small{$1$}}. This allows to conclude that in most of the cases \siS \ has a not isolated and regular critical locus. 

\section{The BV extension for finite spectral triples}
\label{Sect: BV spectral triple}
First stated in \cite{articolo_BV} and then refined in \cite{articolo_ICMP2022}, the notion of {\emph{BV spectral triple}} was introduced with the aim of describing in the language of spectral triples the very rich structure associated by the BV formalism to a gauge theory. However, while being able to translate the structure of a BV-extended theory \swXS\ in the one of a BV spectral triple \BVsp\ had already its relevance, as it allowed to relate two different mathematical languages and structures, still a point was missed: to describe how to perform the BV construction entirely at the level of spectral triples. Answering this last open question is the purpose of this section.\\
\\
To reach this goal, we need to be able to explain how, given as initial spectral triple \isp \ any finite spectral triple on the algebra {\small{$\mathcal{A}_{0}\cong M_{n}(\mathbb{C})$}}, one can extract from it all the needed information to construct an associated BV spectral triple \BVsp. This BV-extension process
$$\begin{array}{ccc}
(\mathcal{A}_{0}, \mathcal{H}_{0}, D_{0}) & ---------\rightarrow  & (\mathcal{A}_{\BV}, \mathcal{H}_{\BV}, D_{\BV}, J_{\BV}) \vspace{-6mm}\\
& \mbox{\tiny{BV-extension process}} & \vspace{3mm}\\
\mbox{\small{initial spectral}} && \mbox{\small{BV spectral}} \vspace{-1mm}\\
\mbox{\small{triple}} & & \mbox{\small{triple}}
\end{array}
$$
 \noindent
 will be presented for any {\small{$\mathrm{U}(n)$}}-affine gauge theory induced by any finite spectral triple on a single matrix algebra, for {\small{$n \geqslant 2$}}. \\
\\
As already hinted by the notation chosen to describe a BV spectral triple, the BV-extension process will force the introduction of a real structure. This can be explained by making a comparison between how, given an initial spectral triple, one associates to it an initial gauge theory \siXS \ and, on the other hand, how a BV spectral triple encodes all the information regarding the ghost spectrum and the extended action. Indeed, the initial configuration space \siX\ is fully determined by the algebra \AI\ as the initial fields are the self-adjoint $1$-forms {\small{$[\Omega^{1}(\mathcal{A}_{0})]_{\text{s.a.}}$}} on the algebra \AI\ while the action \siS\ is given as the spectral action defined by a polynomial {\small{$f \in \mathbb{R}[x]$}}. Interestingly, for the BV case the ghost sector and the BV action will be related to the structure of the spectral triple in the element of the Hilbert space \HBV \ and its associated fermionic action {\small{$S_{\ferm}$}}.\\
\\
 Reasons for these differences can be found in the following observations: first of all, differently than for \siS, in {\small{$S_{\BV}:=\widetilde{S}- S_{0}$}} we do have the appearance of Grassmannian variables. It is precisely the presence of Grassmannian variables what points toward the choice of the fermionic action as the right notion to encode the structure of a BV action \SBV. Then, considering that the fermionic action is defined on (a subspace of) the Hilbert space \HBV, it is natural to expect that ghost/anti-ghost fields enter the BV spectral triple as the components of \HBV. Finally, while in the physical context of the BV-extended theory \swXS\ the bosonic fields are modeled by {\emph{real}} variables, spectral triples are usually defined in a {\emph{complex}} context. Hence, to go from complex to real, we have to introduce a real structure \JBV\ on \HBV \ while the similar passage, from complex to real, for the initial configuration space \siX \ was performed by taking into consideration the $*$-structure on the algebra \AI.\\
 \\
Consequently of what just noticed, the BV-extension process at the level of spectral triples is expected to be a construction of the following type:

$$\begin{array}{ccc}
(\mathcal{H}_{0}, D_{0}) & ---------\rightarrow  & (\mathcal{H}_{\BV}, D_{\BV}, J_{\BV}) \vspace{-6mm}\\
& \mbox{\tiny{BV-extension process}} & \vspace{3mm}\\
\mbox{\small{initial Hilbert sp.}} && \mbox{\small{BV Hilbert sp, BV operator}} \vspace{-1mm}\\
\mbox{\small{and operator}} & & \mbox{\small{and real structure}}
\end{array}
$$
\noindent
where the algebra \ABV \ will not play an explicit role as we will prove that {\small{$\A_{\BV} \cong \A_{0}$}} (cf. Lemma \ref{Lemma: Max algebra}). \\
\\
\noindent 
The last question to answer would then be where in the initial spectral triple we can find the information needed to determine on one hand the number/type of ghost fields in the ghost sector and, on the other, the terms to introduce in the BV-action. The aim of this section is to prove that the BV spectral triple is fully determined by the {\emph{Lie algebra {\small{$\mathfrak{u}(n)$}}}}, in particular by its generators and its structure constants. This implies that all the information needed to determine a BV spectral triple associated to our initial spectral triple are collected in the algebra \AI \ and in its {\emph{algebra of unitaries}}. \\
\\
\noindent
Before proving the main theorem of this section, we will briefly recall the notion of {\emph{BV spectral triple}}.

\begin{definition}
Let \isp\ be a finite spectral triple with induced gauge theory \siXS\ and let \BVsp\ denote an (eventually mixed KO-dimensional) real spectral triple with fermionic action {\small{$S_{\text{ferm}}$}}
$$S_{\text{ferm}}: \mathcal{H}_{\BV,f} \rightarrow \mathcal{H}_{\BV, f} \quad \quad \varphi \mapsto \frac{1}{2}\langle J_{\BV}(\varphi), D_{\BV} \varphi \rangle $$ 
with domain of definition the effective BV Hilbert space {\small{$\mathcal{H}_{\BV, f}\subset \mathcal{H}_{\BV}$}}. Then \BVsp\ 
is a {\emph{BV spectral triple}} associated to \isp\ if {\small{$\mathcal{H}_{\BV, f}$}} can be decomposed as 
$$\mathcal{H}_{\BV, f} \cong \mathcal{Q}_{f}^{*}[1] \oplus \mathcal{Q}_{f},$$
for {\small{$\mathcal{Q}_{f}$}} a {\small{$\mathbb{Z}$}}-graded vector space with shifted dual {\small{$\mathcal{Q}_{f}^{*}[1]$}} and the pair \swXS\ for
$$ \widetilde{X} := ( \mathcal{Q}_{f}^{*}[1] + X_{0}^{*} [1] ) \oplus (X_{0} + \mathcal{Q}_{f}) \quad \mbox{and} \quad  \widetilde{S}[\psi^{*}, \varphi^{*}, \varphi, \psi] := S_{0}[\varphi] + \frac{1}{2}S_{\text{ferm}}[\psi^{*},\psi] $$
is a BV-extended theory associated to \siXS.
\end{definition}
\noindent
{\emph{Note:}} all the variables appearing in {\small{$\mathcal{H}_{\BV, f}$}} as well as in \DBV\ have to be treated as graded variables whose bosonic/fermionic parity coincides with the parity of their degree. This is consistent with the appearance of bosonic/fermionic variables in the action \swS.\\
\\
\noindent
We now define the terms \ABV, \HBV, \DBV\ and \JBV \ separately, postponing to Theorem \ref{Thm: BV spectral triple, all n} the proof that they together determine a BV spectral triple for our {\small{$\mathrm{U}(n)-$}}gauge theories.\\
\\
\noindent
{\bf The algebra {\small{$\mathcal{A}_{\BV}$}}.}
For what concerns the algebra \ABV, as briefly mentioned earlier, it does not play a relevant role and it will be defined only a posteriori: after having determined all the other components of the BV spectral triple, \ABV \ will be taken to be the maximal unital $*$-algebra completing the triple (\HBV, \DBV, \JBV) to a spectral triple satisfying both the commutation rule and the first order condition stated in Definition \ref{def real structure}. By performing a direct computation, in Lemma \ref{Lemma: Max algebra} we will prove that the following identity for our model:
$$\mathcal{A}_{\BV} = \mathcal{A}_{0} = M_{n}(\mathbb{C}).$$
Alternatively, one can reach the same conclusion by applying a method based on the notion of {\emph{Krajewski diagram}} \cite{Krajewski, walter}. The fact that the algebra \AI \ does not change in the BV-extension procedure reflects the fact that \AI \ contains the information about the space of initial fields, which appear both in the initial as well as in the BV-extended theory and that are clearly left unchanged by the BV extension procedure.\\
\\
\noindent
{\bf The Hilbert space {\small{$\mathcal{H}_{\BV}$}}.}
Define the Hilbert space {\HBV} as 
$$\mathcal{H}_{\BV} := \mathcal{Q}^{*}[1] \oplus \mathcal{Q} \quad \mbox{for} \quad \mathcal{Q} = [M_{n}(\mathbb{C})]_{0} \oplus [M_{n}(\mathbb{C})]_{1}$$
where the subscripts {\small{$0, 1$}} appearing in the definition of {\small{$\mathcal{Q}$}} are inserted to keep track of the degree assigned to each {\small{$M_{n}(\mathbb{C})$}}-summand. Its inner product structure is given as usual by the Hilbert--Schmidt inner product on each summand {\small{$M_{n}(\mathbb{C})$}},
$$ \langle \phi, \phi'\rangle = \tr(\phi^{*} \phi'), $$
for {\small{$\phi, \phi' \in M_{n}(\mathbb{C})$}}. Within \HBV, the {\emph{effective Hilbert subspace}} {\small{$\mathcal{H}_{\BV, f}~\subset~\mathcal{H}_{\BV}$}} is the domain of the fermionic action {\small{$S_{\ferm}$}} determined by \DBV:
$$\mathcal{H}_{\BV, f} = \mathcal{Q}_{f}^{*}[1] \oplus \mathcal{Q}_{f} \quad \mbox{for} \quad \mathcal{Q}_{f} = [i\mathfrak{su}(n)]_{0} \oplus [i\mathfrak{su}(n)]_{1},$$
More explicitly, if one fixes a basis for {\small{$i\mathfrak{su}(n)$}} the one given by multiplying by {\small{$i$}} the generalized Gell-Mann matrices {\small{$\{ \sigma_{1}, \dots, \sigma_{n^2 -1}\}$}} together with {\small{$i\sigma_{n^2}:= i \cdot \mathrm{Id}$}}, then a generic vector {\small{$ \varphi \in \mathcal{Q}_{f}$}} has the following form
$$\varphi = \left[ 
\begin{matrix}
[x_{1} & \dots & x_{n^2-1} & 0]_{0} \quad  [C_{1} & \dots & C_{n^2-1} & 0]_{1}
\end{matrix}
\right]^{T} \vspace{1mm}$$
where the notation {\small{$\{x_{p}\}_{p= 1, \dots, n^2-1}$}} is used for the real/bosonic initial field and {\small{$\{C_{q}\}_{q= 1, \dots, n^2-1}$}} are the Grassmannian/fermionic ghost fields of ghost degree $1$. In other words, {\small{$i\mathfrak{su}(n)$}} describes the hermitian traceless part of {\small{$[M_{n}(\C)]_{0}$}} and {\small{$[M_{n}(\C)]_{1}$}}.\\
\\
The reason why the space {\small{$\mathcal{H}_{\BV, f} \subseteq \mathcal{H}_{\BV}$}} just defined is the correct domain for {\small{$S_{\ferm}$}} will be clearer after having introduced the real structure \JBV\ on \HBV\ and the operator \DBV. When determining the subspace {\small{$\mathcal{H}_{\BV, f}$}}, the requirement to fulfill is that none of the terms contributed by the fermionic action cancel each other. Hence, as it will be discussed with more details in Remark \ref{Remark: effective HBV}, the domain {\small{$\mathcal{H}_{\BV, f}$}} of the fermionic action {\small{$S_{\ferm}$}} is taken to be either the subspace preserved by the real structure (up to a product by $i$) or its complementary, depending on the operator \DBV \ anti-commuting or commuting with \JBV, respectively.\\
\\
To conclude, the BV-extension procedure performs the following changes at the level of the BV-Hilbert space:

$$\begin{array}{ccc}
\mathcal{H}_{0} = \mathbb{C}^{n}& ---------\rightarrow  & \mathcal{H}_{\BV} = \mathcal{Q} \oplus \mathcal{Q}^{*}[1] \\
\mbox{\small{initial Hilbert space}} & \mbox{\tiny{+ ghost/anti-ghost fields}} & \mbox{\small{BV-Hilbert space}}
\end{array}
$$
for
$$\mathcal{Q} = [M_{n}(\mathbb{C})]_{0} \oplus [M_{n}(\mathbb{C})]_{1}$$ 
{\emph{Note:}} the ghost sector, that is, the number and type of ghost fields to introduce, is fully determined by the Lie algebra {\small{$\mathfrak{u}(n) = \mathfrak{u}(\mathcal{A}_{0})$}} defined by the unitary elements of \AI: indeed, the ghost fields {\small{$\{C_{q}\}_{q= 1, \dots, n^2-1}$}} can be viewed as the dual basis to the generators {\small{$\{ i\cdot \sigma_{1}, \dots, i\cdot \sigma_{n^2 -1}\}$}} of {\small{$i\cdot \mathfrak{su}(n)$}}, which is obtained from {\small{$\mathfrak{u}(n) = \mathfrak{u}(\mathcal{A}_{0})$}} by quotienting out its center.\\
\\
\noindent
{\bf The operator {\small{$D_{\BV}$}}.}
The self-adjoint linear operator $D_{\BV}$ acting on the Hilbert space $\mathcal{H}_{\BV}$ has the following form
\begin{equation}
\label{DBV in blocks}
D_\BV = \begin{pmatrix} 0 & R \\ R^*& S \end{pmatrix}
\end{equation}
where {\small{$R: \mathcal{Q} \rightarrow \mathcal{Q}^{*}[1]$}} and {\small{$S: \mathcal{Q} \rightarrow \mathcal{Q}$}} are two linear operators represented as follows as block matrices:
{\small{
\begin{gather*}
R:= \frac{1}{2}\begin{pmatrix} 0 & - \mathrm{ad}(C) \vspace{1mm}\\
\mathrm{ad}(C) & -  \mathrm{ad}(X)\\
\end{pmatrix}, \qquad 
S:= \begin{pmatrix} 
0 & \mathrm{ad}(X^{*})\vspace{1mm}\\
\mathrm{ad}(X^{*}) & \mathrm{ad}(C^{*})\vspace{1mm}\\
\end{pmatrix}
\end{gather*}
}}
\noindent
Here {\small{$X$}} and {\small{$X^{*}$}} collectively denote, respectively, the fields {\small{$\{x_{p}\}_{p=1, \dots, n^2}$}} in {\small{$[\Omega^{1}(\mathcal{A}_{0})]_{\text{s.a.}} \cong X_{0}$}} \ and the corresponding antifields {\small{$\{x^{*}_{p}\}_{p=1, \dots, n^2}$}} in {\small{$X_{0}^{*}[1]$}} while {\small{$C$}} and {\small{$C^{*}$}}, similarly, denote the ghost fields in degree {\small{$1$}} and their corresponding anti-ghost fields in degree {\small{$-2$}}. Finally, {\small{$\mathrm{ad}$}} is defined to be the following linear map on {\small{$M_{n}(\mathbb{C})$}}:

\begin{equation*}
\begin{aligned}
\mathrm{ad}(z): M_{n}(\mathbb{C}) &\to M_{n}(\mathbb{C}) ;\\
\phi &\mapsto [\alpha(z_{r}), \phi]_{-},
\end{aligned}
\end{equation*}
where the notation {\small{$[-, -]_{-}$}} is used to denote the matrix commutator, $\alpha(z_{r})$ is a hermitian, traceless $n \times n$-matrix, whose expression in terms of the generalized Gell-Mann matrices is:
$$
\alpha = \tfrac{1}{2}\big[(-z_{1})\sigma_1 + \dots + (-z_{n^2-1})\sigma_{n^2-1}\big]
$$
and the symbol {\small{$z_{r}$}} is used to generically denote a ghost/anti-ghost field. We stress that $\mathrm{ad}$ is a derivation on {\small{$M_n(\C)$}}. Explicitly, using once again the orthonormal basis {\small{$\{ i\cdot \sigma_{1}, \dots, i\cdot \sigma_{n^2 -1}, i\cdot \mathrm{Id}\}$}} for {\small{$M_n(\C)$}} obtained using the generalized Gell-Mann matrices, the linear operator $\mathrm{ad}$, when represented as $n^2 \times n^2$-matrices, has in position {\small{$(p, r)$}} the term {\small{$ - \sum_{q}i \cdot f_{pqr} z_{q}$}}, for {\small{$f_{pqr}$}} denoting the totally anti-symmetric structure constants of {\small{$\mathfrak{u}(n)$}}. Hence, it is straightforward to see that the operator \DBV \ is completely determined by the Lie algebra {\small{$\mathfrak{u}(n) = \mathfrak{u}(\mathcal{A})$}} via its structure constants. \\
\\
\noindent
{\bf The real structure {\small{$\mathbf{J_{\BV}}$}.}}
As mentioned above, the need of going from a structure which is naturally defined in a complex setting, as it is a finite spectral triple, to a physical model, which deals with real fields, forces the introduction of a real structure. We take it to be given by the following anti-linear isometry 
\begin{equation}
\label{Eq: JBV}
J_{\BV}:  \mathcal{H}_{\BV} \rightarrow \mathcal{H}_{\BV}, \quad \mbox{with} \quad J_{\BV}(\varphi_{j}): = i \cdot (\phi_{j}^{*}) 
\end{equation}
for $\phi_{j} \in  [M_{n}(\C)]_{j}$ an element in the direct sum defining \HBV \ and $*$ denoting the matrix adjoint.

\begin{oss}
\label{Remark: effective HBV}
The last remark we make is about the Hilbert space \HBV, the degree/parity of its terms and the way how we defined the effective subspace {\small{$\mathcal{H}_{\BV, f}$}}. Indeed, the condition that the extended action \swS\ has to be a element in {\small{$[\mathcal{O}_{\widetilde{X}}]^{0}$}} implies that {\small{$S_{\text{ferm}} \in [\mathcal{F}(\mathcal{H}_{\BV, f})]^{0}$}}, where {\small{$[\mathcal{F}(\mathcal{H}_{\BV, f})]^{0}$}} denotes the real-valued functional on {\small{$\mathcal{H}_{\BV, f}$}} of total degree $0$. Because of this condition on the total degree of {\small{$S_{\text{ferm}}$}}, real and Grassmannian variables appear both in \DBV\ and {\small{$\mathcal{H}_{\BV, f}$}}. Finally, concerning the form of the effective subspace {\small{$\mathcal{H}_{\BV, f}$}}, as {\small{$\mathrm{ad}$}} has off-diagonal terms of the kind {\small{$i \cdot f_{pqr} z_{r}$}}, that is, only purely imaginary terms, we select the effective part of \HBV\ in order to obtain an action {\small{$S_{\text{ferm}}$}} which has real coefficients. This is why we assign {\small{$\mathcal{H}_{\BV, f}$}} to be the subspace such that
$$J_{\BV}(\mathcal{H}_{\BV, f}) = i \mathcal{H}_{\BV, f}.$$
\end{oss}

\begin{prop}
Let us define 
$$(\A_{\BV}, \H_{\BV}, D_{\BV}, J_{\BV}) := (M_{n}(\C), \mathcal{Q}^{*}[1]  \oplus \mathcal{Q}, D_{\BV}, J_{\BV})$$
with {\small{$\mathcal{Q} = [M_{n}(\mathbb{C})]_{0} \oplus [M_{n}(\mathbb{C})]_{1}$}} and \DBV, \JBV \ as in Equation \eqref{DBV in blocks} and \eqref{Eq: JBV} respectively. Then \BVsp\ is a finite real spectral triple of KO-dimension~$1$.
\end{prop}

\begin{proof}
To check that \BVsp \ define a spectral triple one can observe that \ABV \ is an involutive unital algebra and \HBV \ is a separable Hilbert space. Moreover, one can faithfully represent \ABV \ as bounded operators on \HBV \ by associating to any {\small{$M$}} element in \ABV \ the operator {\small{$P_{M}$}} on \HBV \ given by left multiplication of any summand in \HBV\  by the matrix {\small{$M$}}. As {\small{$M$}} is just a complex matrix, the operator {\small{$P_{M}$}} is clearly bounded. Finally, the operator \DBV \ is clearly self-adjoint on \HBV, due to the structure constants {\small{$f_{pqr}$}} being totally anti-symmetric in all the indices, so that, while in position {\small{$(p,r)$}} one would find the term {\small{$ - \sum_{q}i \cdot f_{pqr} z_{q}$}} in position {\small{$(r,p)$}} one would have {\small{$ - \sum_{q}i \cdot f_{rqp} z_{q} = \sum_{q} i \cdot f_{pqr} z_{q}$}}. The conditions on the domain of \DBV \ being dense, its resolvant being compact and the commutators of \DBV \ with elements in \ABV \ being bounded appearing in Definition \ref{def spectral triple} are automatically satisfied, as we are in the finite dimensional setting. This allows to conclude that {\small{$(\A_{\BV}, \H_{\BV}, D_{\BV})$}} is a finite spectral triple.\\
\\To verify that, by including \JBV \ we obtain a {\emph{real}} spectral triple, we have to check that \JBV \ defines a real structure on {\small{$(\A_{\BV}, \H_{\BV}, D_{\BV})$}} (cf. Definition \ref{def real structure}). With a direct computation one can check that \JBV \ is an antilinear isometry on \HBV \ satisfying {\small{$J^2_{\BV} = \mathrm{Id}$}} and {\small{$J_{\BV}D_{\BV} = - D_{\BV}J_{\BV}$}}, which determines KO-dimension~$1$. Finally, a direct computation, which will be carried out in detail in the proof of Lemma~\ref{Lemma: Max algebra}, shows that also the first-order condition is satisfied. We conclude that \BVsp \ defines a real finite spectral triple of KO-dimension $1$.
\end{proof}

\begin{lemma}
\label{Lemma: Max algebra}
Let \HBV, \JBV\ and \DBV\ be as defined above. Then the maximal unital subalgebra \ABV {\small{$ \subset \mathcal L(\H_\BV)$}} of linear operators on \HBV \ that satisfies
$$
[a,J_\BV b^* J^{-1}_\BV]= 0, \qquad [[D_\BV,a],J_\BV b^* J^{-1}_\BV]= 0; \qquad (a,b\in \mathcal{A}_{\BV})
$$
is given by {\small{$\mathcal{A}_{\BV} = M_n(\C)$}} acting diagonally on \HBV. 
\end{lemma}
\proof
The commutation rule {\small{$[a,J_\BV b^* J_{\BV}^{-1}]= 0$}} for all {\small{$a,b \in \A_{\BV}$}} implies that \HBV \ carries an \ABV -bimodule structure. This restricts \ABV\  to be a subalgebra in
$$\mathcal{A}_{\BV} \subset [M_{n}(\C)]^{\oplus 4}$$
acting diagonally on {\small{$\mathcal{H}_{\BV} = \mathcal{Q} \oplus \mathcal{Q}^{*}[1]$}}. Then, by a straightforward computation of the double commutator 
$$\left[[D_{\BV},(a_1,\dots, a_{4})], J_\BV (b_1, \dots, b_{4})J_\BV^{-1} \right]$$ 
it follows that the first-order condition implies {\small{$a_1= \dots = a_{4}$}} and {\small{$b_1= \dots = b_{4}$}}. This selects the subalgebra {\small{$M_n(\C)$}} as the maximal one for which both of the above conditions are satisfied. Alternatively, one can prove the statement using a graphical method based on the notion of {\emph{Krajewski diagram}} \cite{Krajewski, walter}. 
\endproof
\noindent
We conclude the section by stating its main theorem: the method described above to associate to a given initial finite spectral triple {\small{$(\mathcal{A}_{0}, \mathcal{H}_{0}, D_{0})$}}, a corresponding real spectral triple {\small{$(\mathcal{A}_{\BV}, \mathcal{H}_{\BV}, D_{\BV}, J_{\BV})$}} of KO-dimension $1$ (mod 8) accomplishes our goal of encoding the construction of a BV-extension for our {\small{$\mathrm{U}(n)$}}-gauge theory all within the framework of noncommutative geometry.  

 \begin{theorem}
 \label{Thm: BV spectral triple, all n}
Let {\small{$(\mathcal{A}_{0}, \mathcal{H}_{0}, D_{0}) = (M_{n}(\mathbb{C}), \mathbb{C}^{n}, D_{0})$}} be a finite spectral triple with induced gauge theory \siXS. Then an associated BV-spectral triple is given by:
$$(\mathcal{A}_{\BV}, \mathcal{H}_{\BV}, D_{\BV}, J_{\BV}) = (M_{n}(\mathbb{C}), \mathcal{Q} \oplus \mathcal{Q}^{*}[1], D_{\BV}, J_{\BV})$$ 
for {\small{$\mathcal{Q} = [M_{n}(\mathbb{C})]_{0} \oplus [M_{n}(\mathbb{C})]_{1}$}} and \DBV, \JBV\  as defined in \eqref{DBV in blocks} and \eqref{Eq: JBV}, respectively. 
 \end{theorem}

\begin{proof}
To prove \BVsp \ defines a BV spectral triple for \isp \ one has to check that the pair
$$ \widetilde{X} := (\mathcal{Q}_{f} + X_{0}) \oplus (X_{0}^{*} [1] + \mathcal{Q}_{f}^{*}[1]) \hspace{10mm}  \widetilde{S} := S_{0} +\frac{1}{2} S_{\text{ferm}} $$
with 
$$S_{\text{ferm}}: \mathcal{H}_{\BV,f} \rightarrow \mathcal{H}_{\BV, f} \quad \quad \varphi \mapsto \frac{1}{2}\langle J_{\BV}(\varphi), D_{\BV} \varphi \rangle $$ 
is a BV-extended theory associated to \siXS, for {\small{$X_{0} \cong [\Omega^{1}(\A_{0})]_{s.a.}$}} and {\small{$S_{0}[\varphi] = \text{Tr}(f(D_{0}+ \varphi))$}}. Explicitly, for the extended configuration space \swX \ we have to prove that \swX \ is a {\small{$\mathbb{Z}$}}-graded super vector space, with {\small{$[\widetilde{X}]^{0} = X_{0}$}} and such that {\small{$\widetilde{X} \cong \mathcal{F} \oplus \mathcal{F}^{*}[1]$}}, for {\small{$\mathcal{F} = \oplus_{i \geqslant 0} \mathcal{F}^{i}$}} is a graded locally free {\small{$\mathcal{O}_{X_{0}}$}}-module with homogeneous components of finite rank, for {\small{$\mathcal{O}_{X_{0}}$}} the algebra of regular functions on {\small{$X_{0}$}}. All this immediately follows from the explicit definition of \HBV, where the {\small{$\mathcal{O}_{X_{0}}$}}-module structure can be naturally induced. \\
\\For what concerns the extended action \swS, we have to prove that {\small{$\widetilde{S} \in [\mathcal{O}_{\widetilde{X}}]^{0}$}}, is a real-valued regular function on {\small{$\widetilde{X}$}}, with {\small{$\widetilde{S}|_{X_{0}}=S_{0}$}}, {\small{$\widetilde{S}\neq S_{0}$}} and such that it solves the classical master equation {\small{$
\{\widetilde{S}, \widetilde{S}\}=0,$}} where {\small{$\{ -, -\}$}} denotes the graded Poisson structure on the graded algebra {\small{$\mathcal{O}_{\widetilde{X}}$}}. 
While the degree and regularity conditions can be straightforwardly checked, we focus on proving that \swS \ is a solution of the classical master equation. To do so, let us compute {\small{$S_{\text{ferm}}$}} explicitly in terms of the variables {\small{$x_{p}, x^{*}_{p}, C_{q}, C^{*}_{q}$}}. Taking into consideration the parity of the variables as well as the fact that the structure constants {\small{$f_{pqr}$}} of the Lie algebra {\small{$\mathfrak{u}(n)$}} are totally anti-symmetric in the three indices, one can prove that:
\begin{equation}
\label{eq: extended action U(n)}
\widetilde{S}= S_{0} + \frac{1}{2}S_{\ferm} = S_{0} + \sum_{p, q,r} f_{pqr} \big[\ x_{p}^{*}x_{q}C_{r} + \frac{1}{2}C^{*}_{p}C_{q}C_{r}\big], 
\end{equation}
To conclude the proof we have to verify that {\small{$\{S_{0} + 1/2 S_{\ferm}, S_{0} + 1/2 S_{\ferm}\} =0$}}. As \siS \ is itself a solution to the classical master equation, we have to prove that {\small{$\big\{S_{0}, S_{\ferm}\big\} + \frac{1}{2}\big\{S_{\ferm}, S_{\ferm}\big\} =0$}}. \\
\\
Analyzing separately the contributions of the different terms, one obtains:\\
\\
$(1)$\ {\small{$\big\{S_{0}, S_{\ferm}\big\} = \sum_{pqr} f_{pqr} \ \partial_{p} S_{0} \ x_{q} C_{r} = \sum_{r} \big[ \sum_{pq} f_{pqr} \ \partial_{p} S_{0} \ x_{q} \big] C_{r} =0$ }} \\
\\
 where the last identity follows from noticing that, for any fixed $r$ the sum in brackets is zero as direct consequence of the action functional \siS\ being invariant under the adjoint action of {\small{$\mathfrak{u}(n)$}}.
 \noindent
 {\small{
\begin{multline*}
(2) \ \sum_{p, q, r, s, a}\Big[\big\{ f_{par} \ x^{*}_{p}x_{a}C_{r}, f_{aqs} \ x^{*}_{a}x_{q}C_{s} \big\} + \big\{ f_{pas} \ x^{*}_{p}x_{a}C_{s}, f_{aqr} \ x^{*}_{a}x_{q}C_{r} \big\} \  + \\
\big\{ f_{pqr} \ x^{*}_{p}x_{q}C_{a}, \frac{1}{2} f_{ars} \ C^{*}_{a}C_{r}C_{s} \big\} \Big] = \sum_{p, q, r, s} \big[ \sum_{a} (f_{apr}f_{aqs} - f_{aps}f_{aqr} - f_{apq}f_{ars}) \big] x^{*}_{p}x_{q}C_{r}C_{s} =0 
\end{multline*} 
}} 
where the last passage follows from noticing that, for every fixed set of indices $p, q, r, s$, the expression in the bracket involving thee structure constants is zero as consequence of the Jacobi identity satisfied by the Lie bracket in a Lie algebra.
\noindent
{\small{
\begin{multline*}
(3) \ \sum_{p, q, r, s, a}\Big[\big\{ f_{pqa} \ C^{*}_{p}C_{q}C_{a}, f_{ars} \ C^{*}_{a}C_{r}C_{s} \big\} + \big\{ f_{pra} \ C^{*}_{p}C_{r}C_{a}, f_{aqs} \ C^{*}_{a}C_{q}C_{s} \big\} \ + \\
\big\{ f_{psa} \ C^{*}_{p}C_{s}C_{a}, f_{aqr} \ C^{*}_{a}C_{q}C_{r} \big\} \Big] = \sum_{p, q, r, s} \big[ \sum_{a} (-f_{apq}f_{ars} + f_{apr}f_{aqs} - f_{aps}f_{aqr}) \big] C^{*}_{p}C_{q}C_{r}C_{s} =0 
\end{multline*}
}}
\noindent
Similarly as above, also in this last case the sum is zero as a consequence of the Jacobi identity. This allows to conclude that the functional \swS \ induced by \BVsp \ is indeed a solution to the classical master equation and \BVsp \ is then a BV spectral triple associate to our initial spectral triple \isp.
\end{proof}

\noindent
{\emph{Note:}} we conclude by remaking that, given an initial spectral triple \isp, the only information needed to determine the associated BV spectral triple is {\small{$\mathfrak{u}(\mathcal{A})$}}, that is, the Lie algebra of the unitary elements of \AI.

\section{The BV complex as Hochschild complex}
\label{Sect: BV and Hochschild complex}
\noindent
We concluded the previous section by describing the BV spectral triple associated to any {\small{$\mathrm{U}(n)$}}-gauge theory induced by a finite spectral triple on the algebra {\small{$M_{n}(\C)$}}. We also saw how the BV spectral triple incorporates all the information of the BV-extended pair \swXS. Hence, one can reasonably expect to be able to define the \cBRST\ complex straight from the BV spectral triple. To make this possible, first of all we need to relate this BV complex to another cohomology complex which appears naturally in the context of spectral triples. \\
\\
As indicated in the title, the main result presented in this section is the construction of a graded isomorphism 
$$\Phi: \mathcal{C}_{\BVt}^{\bullet}(\widetilde{X}, d_{\widetilde{S}}) \longrightarrow \mathcal{C}_{\Hoch, \Delta}^{\bullet}(\mathcal{B},  \mathcal{M})$$ 
between the BV cohomology complex {\small{$\mathcal{C}_{\BVt}^{\bullet}(\widetilde{X}, d_{\widetilde{S}})$}} induced by a BV-extended theory \swXS\ and the Hochschild complex {\small{$\mathcal{C}_{\Hoch}^{\bullet}(\mathcal{B},  \mathcal{M})$}} of a coalgebra  \scB\ over a comodule \scM, for a suitable choice of the pair {\small{$(\mathcal{B}, \mathcal{M})$}}. 
The construction will be presented in the general context of finite spectral triple, for a generic BV spectral triple \BVsp\ associated to an initial spectral triple \isp. In this context, constructing an isomorphism of cochain complexes entails to define, for any $k \in \mathbb{Z}$, an isomorphism of finitely generated {\small{$\mathcal{O}_{X_{0}}$}}-modules 
$$\Phi^{k}: \mathcal{C}_{\BVt}^{k}(\widetilde{X}, d_{\widetilde{S}}) \longrightarrow \mathcal{C}_{\Hoch}^{k}(\mathcal{B},  \mathcal{M})$$
between the degree $k$ cochain spaces, which is compatible with the two corresponding coboundary operators, that is, which satisfies the condition
$$d_{H} \circ  \Phi^{k}|_{\varphi} = \Phi^{k+1} \circ d^{k}_{\widetilde{S}}|_{\varphi}$$
for all {\small{$\varphi \in \mathcal{C}_{\BVt}^{k}(\widetilde{X}, d_{\widetilde{S}})$}}. To reach this goal, we first recall the definition of the Hochschild cohomology complex for coalbegras over comodules. Then we explain how, given a BV spectral triple \BVsp \ one determines the pair {\small{$(\mathcal{B},  \mathcal{M})$}} of a suitable coalgebra and comodule. Finally, we define the collection of maps $\{ \Phi^{k}\}_{k \in \mathbb{Z}}$, proving they give the complex isomorphism we are looking for. For more details on coalgebras and their cohomology complexes a classical reference is \cite{Kassel_quantum_Groups}.

\begin{definition}
A {\emph{coalgebra}} is a triple {\small{$(\B, \Delta, \epsilon )$}} where \scB \ is a vector space over a field {\small{$\mathbb{K}$}} and {\small{$\Delta: \B \rightarrow \B \otimes \B$}} \ and {\small{$\epsilon: \B \rightarrow \mathbb{K}$}} are linear maps satisfying the coassociativity and the counit axioms, respectively:
$$(\mathrm{id}_{\B} \otimes \Delta) \circ \Delta = (\Delta \otimes \mathrm{id}_{\B}) \circ \Delta \quad \quad (\mathrm{id}_{\B} \otimes \epsilon) \circ \Delta = \mathrm{id}_{\B} = (\epsilon \otimes \mathrm{id}_{\B}) \circ \Delta. $$
\end{definition}

\begin{definition}
    Consider two coalgebras {\small{$(\B, \Delta, \epsilon )$}} and {\small{$(\B^{\prime}, \Delta^{\prime}, \epsilon^{\prime} )$}}. A linear map {\small{$f: \B \rightarrow\B^{\prime}$}} \ is a {\emph{morphism of coalgebras}} if:
    $$(f \otimes f) \circ \Delta = \Delta^{\prime} \circ f \quad \quad \epsilon = \epsilon^{\prime} \circ f.$$
\end{definition}

\begin{definition}
Let {\small{$(\B, \Delta, \epsilon )$}} be a coalgebra. Then a {\emph{right module}} over \scB \ is a pair {\small{$(\mathcal{M}, \omega_{R})$}} where \scM \ is a vector space and {\small{$\omega_{R}: \mathcal{M} \rightarrow \M \otimes \B$}} is a linear map, called the {\emph{coaction}} of \scB \ on \scM, such that the following axioms are satisfied:
$$ (\mathrm{id}_{\M} \otimes \Delta) \circ \omega_{R} = ( \omega_{R} \otimes \mathrm{id}_{\B}) \circ \omega_{R} \quad \mbox{and} \quad (\mathrm{id}_{\M}\otimes\epsilon) \circ \omega_{R} = \mathrm{id}_{\M}.$$

The notion of Hochschild cohomology was first introduced by Hochschild \cite{Def_Hoch_Cohom, Def_Hoch_Cohom2} while its presentation in terms of derived functor is due to Cartan and Eilenberg \cite{homol}. Because we aim to determine a direct and explicit relation between the BV and the Hochschild complexes for our model, we are interested only in the realization of the Hochschild complex as determined by the so-called standard resolution of an algebra in the category of A-bimodules (cf. \cite{Kassel}). To obtain a notion of the Hochschild complex in the coalgebra case one simply considers this more explicit definition and dualizes it. 

\begin{definition}
Let {\small{$(\B, \Delta, \epsilon )$}} be a coalgebra  and let {\small{$(\M, \omega_{R})$}} be a right comodule over \scB. Then the {\emph{Hochschild cohomology complex}} of the coalgebra \scB \ over the comodule \scM \ has cochain spaces and coboundary operator defined as
follows
$$ \mathcal{C}^{p}_{H, \Delta}(\B, \M):= \M \otimes \B^{\otimes p} \quad \quad d_{H, \Delta}:\mathcal{C}^{p}_{H, \Delta}(\B, \M) \rightarrow \mathcal{C}^{p+1}_{H, \Delta}(\B, \M) $$
with 
$$ d_{H, \Delta}(\varphi):= \omega_{R}(f) \otimes y_{1} \otimes \dots \otimes y_{p} + \sum_{j=1}^{p}(-1)^{j} f \otimes y_{1} \otimes \dots \otimes \Delta(y_{j}) \otimes \dots \otimes y_{p}  $$
for {\small{$\varphi = f \otimes y_{1} \otimes \dots \otimes y_{p} \in \mathcal{C}^{p}_{H, \Delta}(\B, \M)$}}. The definition of the coboundary map {\small{$d_{H, \Delta}$}} can then be extended to the entire vector space of cochains of degree $p$ by linearity.
\end{definition}
\noindent
\begin{lemma}
\label{Lemma: Hoch is complex, non graded}
Given {\small{$(\B, \Delta, \epsilon )$}} a coalgebra  and {\small{$(\M, \omega_{R})$}} a right comodule over \scB, the pair {\small{$(\mathcal{C}^{\bullet}_{H, \Delta}(\B, \M)), d_{H, \Delta}$)}} defines a cohomology complex.
\end{lemma}

\begin{proof}
    Given {\small{$\mathcal{C}^{\bullet}_{H, \Delta}(\B, \M)$}} the vector space {\small{$\M \otimes \B^{\otimes p}$}}, the map {\small{$d_{H, \Delta}$}} is a well-defined linear map of degree {\small{$1$}}. To conclude that {\small{$(\mathcal{C}^{\bullet}_{H, \Delta}(\B, \M)), d_{H, \Delta}$)}} defines a cohomology complex the last thing to check is that {\small{$d_{H, \Delta}$}} is a coboundary operator, that is, it satisfies {\small{$(d_{H, \Delta})^{2} =0$}}. Because of the linearity of the map {\small{$d_{H, \Delta}$}} we can verify this coboundary condition on one single term {\small{$\varphi \in \mathcal{C}^{p}_{H, \Delta}(\B, \M)$}}. Due once again to the linearity of {\small{$d_{H, \Delta}$}}, without loss of generality we can assume that {\small{$\omega_{R}(f)= m \otimes x$}}, with {\small{$m \in \M$}} and {\small{$x \in \B$}}, for $f$ the coefficient in \scM\ of the element $\varphi$ taken into consideration. 
    
    \begin{align*}
    & d_{H, \Delta} \big( d_{H, \Delta} (\varphi) \big)  \\     & \quad = d_{H, \Delta} \big(\omega_{R}(f) \otimes y_{1} \otimes \dots \otimes y_{p} + \sum_{j=1}^{p}(-1)^{j} f \otimes y_{1} \otimes \dots \otimes \Delta(y_{j}) \otimes y_{j+1} \otimes \dots \otimes y_{p} )\\
   & \quad 
   = \omega_{R}(m) \otimes x \otimes y_{1} \otimes \dots \otimes y_{p} - m \otimes \Delta(x) \otimes y_{1} \otimes \dots \otimes y_{p} \\
   & \quad \quad + \sum_{j= 1}^{p}(-1)^{j}m \otimes x \otimes \dots \otimes \Delta(y_j) \otimes \dots \otimes y_{p}\\
   & \quad \quad + \sum_{i= 1}^{p}(-1)^{i+1}m \otimes x \otimes \dots \otimes \Delta(y_i) \otimes \dots \otimes y_{p} \\
   & \quad \quad + \sum_{i=1}^{p}(-1)^{i} \sum_{j=1}^{i-1} (-1)^{j+1} m \otimes \dots \otimes \Delta(y_{j})\otimes \dots \otimes \Delta(y_i) \otimes \dots \otimes y_{p}\\
   & \quad \quad + \sum_{i=1}^{p}(-1)^{i} \sum_{j=i+3}^{p} (-1)^{j} m \otimes \dots \otimes \Delta(y_i) \otimes \dots \otimes \Delta(y_{j}) \otimes \dots \otimes y_{p}\\
   & \quad \quad + \sum_{i= 1}^{p}(-1)^{2i}m \otimes \dots \otimes \big(\Delta \otimes \mathrm{Id}\big)(\Delta(y_{i})) \otimes \dots \otimes y_{p}\\
   & \quad \quad + \sum_{i= 1}^{p}(-1)^{2i+1}m \otimes \dots \otimes \big( \mathrm{Id} \otimes \Delta \big)(\Delta(y_{i})) \otimes \dots \otimes y_{p} =0
        \end{align*}
To reach the final result it is enough to observe that the first two summands cancel each other due to the compatibility of the coaction with the coproduct. For what concerns the third/four and fifth/sixth summands, they cancel pairwise having opposite sign. The last two summands, cancel due to the coassociativity of the coproduct. This allows to conclude that the pair {\small{$(\mathcal{C}^{\bullet}_{H, \Delta}(\B, \M)), d_{H, \Delta}$)}} defines a cohomology complex.
\end{proof}
\end{definition}

This definition of a Hochschild cohomology complex for a coalgebra over a comodule can be extended to the case of a graded coalgebra. We need to consider this generalisation due to the presence in the BV complex of generators of degree different from {$0, 1$}.

\begin{definition}
A {\emph{1-shifted (non-unital) graded coalgebra}} is a pair {\small{$(\B, \Delta)$}} such that $\B$ is a {\small{$\Z$}}-graded vector space, {\small{$\mathcal{B} = \bigoplus_{n \in \mathbb{Z}} \mathcal{B}_{n}$}}, and {\small{$\Delta: \B \rightarrow \B \otimes \B$}} is a linear map such that
$$\Delta(y_{a}) \in \bigoplus_{i+j = a+1} \B_{i} \otimes \B_{j},$$
for {\small{$y_{a}$}} a homogeneous element in \scB \ of degree $a$. Moreover, the coproduct {\small{$\Delta$}} has to satisfy the following graded coassociativity condition:
$$(\Delta \otimes \mathrm{id})(\Delta(\varphi)) = (-1)^{|z^{(1)}| +1} (id \otimes \Delta) (\Delta(\varphi)), $$
for $\Delta(\varphi) = z^{(1)} \otimes z^{(2)}$ and for $|z^{(1)}|$ denoting the degree of the homogenous term $z^{(1)}$ in the graded vector space $\B$.
\end{definition}

Similarly, we need to adapt to this graded context also the notion of right comodule.

\begin{definition}
Let {\small{$(\B, \Delta)$}} be a 1-shifted graded coalgebra. Then a degree-1 right comodule over $\B$ is a pair {\small{$(\M, \omega)$}} with \scM \ a vector space and {\small{$\omega_{R}$}} a degree 1 coaction, that is, a linear map {\small{$\omega_{R}: \M \rightarrow \M \otimes \B_{1}$}} such that it satisifes a graded compatibility with the coproduct:
$$(\omega_{R} \otimes \mathrm{id}) \circ \omega_{R} = - (id \otimes \Delta) \circ \omega_{R}.$$
\end{definition}

\begin{definition}
\label{def: Hochschild cohom, graded}
Let {\small{$(\B, \Delta)$}} be a 1-shifted graded coalgebra, with decomposition {\small{$\B = \oplus_{n \in \mathbb{Z}} B_{n}$}}, and let {\small{$(\M, \omega_{R})$}} be a degree-1 right comodule over $\B$. Then the {\emph{graded Hochschild complex}} of the coalgebra \scB \ on the comodule \scM \ has the cochain spaces defined as
$$\mathcal{C}^{p}_{H, \Delta}(\B, \M): = \M \otimes T^{p}(\B)$$
where {\small{$T^{p}(\B)$}} denotes the super-graded tensor algebra of $\B$ with
\begin{multline*}
    T^{p}(\B):= \{\varphi = \sum_{I} y_{1}^{i_{1}} \otimes \dots \otimes y_{k}^{i_{k}}, \ y^{a} \in \B_{a}, \ I=(i_{1}, \dots, i_{k}), \ i_{1} + \dots + i_{k} =p \} / \sim
\end{multline*}
for $p \in \mathbb{Z}$, where the quotient is taken with respect to the equivalence relations implementing the graded commutativity of the different cofactors in {\small{$\varphi$}}. The coboundary operator {\small{$d_{H, \Delta}:\mathcal{C}^{p}_{H, \Delta}(\B, \M) \rightarrow \mathcal{C}^{p+1}_{H, \Delta}(\B, \M) $}}, it is defined as 
\begin{multline}
\label{Eq: def hoch coboundary graded}
    d_{H, \Delta}(\varphi):= \omega_{R}(f) \otimes y^{i_{1}}_{1} \otimes \dots \otimes y^{i_{k}}_{k}\\
    + \sum_{j=1}^{k}(-1)^{i_{1} + \dots + i_{j-1}} f \otimes y^{1}_{1} \otimes \dots \otimes \Delta(y^{i_{j}}_{j}) \otimes y^{i_{j+1}}_{j+1} \otimes \dots \otimes y^{i_{k}}_{k}  
\end{multline} 
for {\small{$\varphi = f \otimes y^{i_{1}}_{1} \otimes \dots \otimes y^{i_{k}}_{k} \in \mathcal{C}^{p}_{H, \Delta}(\B, \M)$}}.
\end{definition}

\begin{lemma}
Given {\small{$(\B, \Delta)$}} a 1-shifted graded coalgebra  and {\small{$(\M, \omega_{R})$}} a  degree-1 right comodule over \scB, the pair {\small{$(\mathcal{C}^{\bullet}_{H, \Delta}(\B, \M)), d_{H, \Delta}$)}} defines a cohomology complex.
\end{lemma}

\begin{proof}
Before starting the proof we notice that the coboundary operator {\small{$d_{H, \Delta}$}} is well defined on the quotient due to the grading of the coproduct {\small{$\Delta$}}. Similarly to Lemma \ref{Lemma: Hoch is complex, non graded}, to prove the statement we only have to check that the linear operator {\small{$d_{H, \Delta}$}} as defined in \eqref{Eq: def hoch coboundary graded} satisfies the coboundary condition, that is, {\small{$d_{H, \Delta}^{2} = 0$}}. This can be checked with a direct computation. The different summands appearing will cancel in pairs, with the properties used to draw this conclusion being, in order, the graded compatibility of the coaction {\small{$\omega_{R}$}} with the coproduct {\small{$\Delta$}} and the graded coassociativity of {\small{$\Delta$}}. Finally, useful remarks would be to notice that {\small{$|\omega_{R}(f)| = 1$}} and {\small{$|\Delta(y^{i_{j}}_{j})| = i_{j} +1$}}.
\end{proof}

After having recalled the definition of the Hochschild complex in the coalgebra setting, we have now to determine a pair {\small{$(\mathcal{B}, \mathcal{M})$}} such that their Hochschild complex {\small{$(\mathcal{C}^{\bullet}_{\Hoch, \Delta}(\mathcal{B}, \mathcal{M}), d_{\Hoch,\Delta})$}} coincides with the BV complex {\small{$(\mathcal{C}^{\bullet}_{\BV}(\widetilde{X}, d_{\widetilde{S}}), d_{\widetilde{S}})$}}. \\
\\
\noindent
{\bf The coalgebra $\mathcal{B}$.} Let \BVsp \ denote a BV spectral triple associated to an initial spectral triple \isp. Then let $\B = \oplus_{n\in \mathbb{Z}} \B_{n}$ be given, as vector space, by:
\begin{equation}
\label{Eq: def B}
\mathcal{B}:=
(\mathcal{Q}_{f} + [\mathcal{O}_{X_{0}}]^{\leq(\deg(S_{0})-1)}) \oplus (X_{0}^{*} [1] + \mathcal{Q}_{f}^{*}[1]).    
\end{equation}
More explicitly, the homogeneous components of \scB \ are 
$$\B_{0} = [\mathcal{O}_{X_{0}}]^{\leq(\deg(S_{0})-1)}, \quad \B_{-1} = X_{0}^{*} [1], \quad \B_{m} = [\mathcal{Q}_{f}]^{m}, \quad \B_{-m} = [\mathcal{Q}_{f}^{*}[1]]^{-m}$$ for {\small{$m>1$}}, where {\small{$[\mathcal{O}_{X_{0}}]^{\leq(\deg(S_{0})-1)}$}} denotes the finite dimensional vector space of polynomials in the initial fields \siX \ and with degree strictly lower then the degree of the action functional \siS.\\
\\
To complete the description of the coalgebra \scB, we still have to endow it with a coproduct structure {\small{$\Delta: \B \rightarrow \B \otimes \B$}} such that {\small{$\Delta$}} is a linear map satisfying 
$$\Delta(y^{a}) \in \bigoplus_{i+j = a+1} \B_{i} \otimes \B_{j}.$$ 
Given {\small{$y_{a}$}} a homogeneous element in \scB, we define 
\begin{equation}
\label{Eq: def Delta}
\Delta(y^{a}):= \left\lbrace S_{0} + \frac{1}{2} S_{\ferm}, y^{a} \right\rbrace     
\end{equation}

for {\small{$\{ - , - \}$}} the Poisson bracket structure defined on {\small{$\mathcal{O}_{\widetilde{X}} \cong \mathcal{F}(\mathcal{H}_{\BV})$}}, for {\small{$\mathcal{F}(\mathcal{H}_{\BV})$}} denoting the space of polynomials on {\small{$\mathcal{H}_{\BV, f}$}}, and induced by the pairing between fields/ghost fields and corresponding antifields/anti-ghost fields.

\begin{lemma}
    \label{lemma: B is a graded coalgebra}
The pair {\small{$(\B, \Delta)$}} defined in \eqref{Eq: def B} and \eqref{Eq: def Delta} is a 1-shifted graded coalgebra.
\end{lemma}

\begin{proof}
To prove that {\small{$(\B, \Delta)$}} defines a graded coalgebra it is enough to notice that, by construction, the map {\small{$\Delta$}} is linear as the Poisson bracket is linear in both entries. Concerning the condition on the degree of the coproduct imposing that {\small{$\Delta(y_{a})$}} should be an element in {\small{$\bigoplus_{i+j = a+1} \B_{i} \otimes \B_{j}$}}, let us first discuss it for {\small{$y^{a} \in \B_{a}$}} with {\small{$a \neq -1$}}. In this case the only not-trivial contribution to the Poisson bracket comes from the contraction of {\small{$y^{a}$}}  with {\small{$S_{\ferm}$}}. By definition, the fermionic action is exactly cubic in the field/antifields/ghost/anti-ghost fields. A consequence of this property is that, when contracting {\small{$S_{\ferm}$}} with an element {\small{$y^{a}$}}, the result is always an element in the tensor product {\small{$\B_{i} \otimes \B_{j}$}}, with {\small{$i + j = a+1$}} that follows from the fact that the term {\small{$ S_{0} + 1/2 \cdot S_{\ferm}$}} has degree $0$ and the Poisson bracket has degree 1. If instead we consider a homogeneous element {\small{$y_{i}^{a}= x^{*}_{i}$}} of degree {\small{$a = -1$}}, in addition to the contribution coming from the contraction with the fermionic action, we would have also a term of the kind {\small{$\partial S_{0}/\partial x_{i}$}} coming from the contraction of {\small{$x^{*}_{i}$}} with the action \siS. However, also this term can be viewed as an element in  {\small{$\B_{0} \otimes \B_{0}$}}, with the condition on the degree clearly satisfied. Finally, to conclude that {\small{$(\B, \Delta)$}} is a graded coalgebra we still have to verify that {\small{$\Delta$}} satisfies the graded coassociativity condition:
$$(\Delta \otimes \mathrm{id}) \circ \Delta(\varphi) = (-1)^{|z^{(1)}| +1} (id \otimes \Delta)(\Delta(\varphi)), $$
for {\small{$\Delta(\varphi) = z^{(1)} \otimes z^{(2)}$}}. Substituting the explicit expression of {\small{$\Delta$}} in the left-hand side of the above equality, one would find:
$$(\Delta \otimes \mathrm{id}) \circ \Delta(\varphi) = \{ \widetilde{S}, z^{(1)} \} \otimes z^{(2)}, $$
where we use the fact that {\small{$S_{0} + 1/2 \cdot S_{\ferm} = \widetilde{S}$}} by the assumption made on the BV spectral triple \BVsp. Similarly, the right-hand side would be:
$$(id \otimes \Delta)(\Delta(\varphi)) = z^{(1)} \otimes \{ \widetilde{S}, z^{(2)}\}.$$
However, we know that {\small{$\{ \widetilde{S}, - \}$}} defines a coboundary operator, that is, given any element {\small{$\varphi$}}, it holds that {\small{$\{ \widetilde{S}, \{ \widetilde{S}, \varphi\} \} =0$}}. Explicitly, this determines the following series on equalities:
$$ 0 = \{ \widetilde{S}, \{ \widetilde{S}, \varphi\} \} = \{ \widetilde{S}, z^{(1)} \otimes z^{(2)} \} =  \{ \widetilde{S}, z^{(1)} \}\otimes z^{(2)} + (-1)^{|z^{(1)}|} z^{(1)}\otimes  \{ \widetilde{S}, z^{(2)} \}, $$
which implies the graded coassociativity of {\small{$\Delta$}}. 
\end{proof}

\noindent
{\bf The comodule $\mathcal{M}$.} Similarly to what we have done for the definition of the coalgebra \scB, let \BVsp \ be a BV spectral triple associated to an initial spectral triple \isp. Then we denote by \scM \ be the algebra generated by the space of self-adjoint 1-forms {\small{$\Omega^{1}(\mathcal{A}_{0})_{\text{s.a.}}$}}
\begin{equation}
    \label{eq: def M}
\mathcal{M}:= \langle \Omega^{1}(\mathcal{A}_{0})_{\text{s.a.}} \rangle\cong \mathcal{O}_{X_{0}},
\end{equation}

where we recall that
$$\Omega^{1}(\mathcal{A}_{0})_{\text{s.a.}} := \big\{ \varphi = \sum_{i} a_{j} [D_{0}, b_{j}] : \varphi^{*}= \varphi, a_{j}, b_{j} \in \mathcal{A}_{0}\big\} \cong X_{0}.$$
To complete \scM \ to a degree-1 right comodule over \scB \ we need to endow it with a linear map {\small{$\omega_{R}: \M \rightarrow \M \otimes \B_{1}$}} satisfying the graded compatibility condition with the coproduct {\small{$\Delta$}} on \scB. Given an element {\small{$f \in \M$}}, we define:
\begin{equation}
\label{Eq: def of omega}    
\omega_{R}(f):= \left\lbrace S_{0} + \frac{1}{2} S_{\ferm}, f\right\rbrace,
\end{equation}
with {\small{$\{ -, - \}$}} once again denoting the Poisson structure induced by the pairing of fields/antifields.

\begin{lemma} 
The pair {\small{$(\M, \omega_{R})$}}, with {\small{$\omega_{R}$}} defined in \eqref{Eq: def of omega} is a degree-1 right comodule over the coalgebra {\small{$(\B, \Delta)$}}.
\end{lemma}

\begin{proof}
To conclude that {\small{$(\M, \omega_{R})$}} is a degree-1 right comodule we only have to check that the map {\small{$\omega_{R}$}} takes values, as claimed, in the tensor product {\small{$\M \otimes \B_{1}$}}, it is linear and, finally, it satisfies a graded compatibility condition with the coproduct on 
\scB. However, when contracting an element $f$ in \scM \ with the action \swS, the only non-zero contribution comes from the terms in {\small{$S_{\ferm}$}} containing the antifields {\small{$x^{*}_{a}$}}, which have ghost degree $-1$. As already observed, the action {\small{$S_{\ferm}$}} is exactly cubic in the variables. Even more, when the gauge theory considered has the algebra of gauge symmetries which is closed on-sell one has an action {\small{$S_{\ferm}$}} which is linear in the antifields/anti-ghost fields and hence quadratic in the fields/ghost fields. This is in particular what happens in the {\small{$\mathrm{U}(n)$}}-models we analyzed in Section \ref{Sect: BV spectral triple}. All this, together with the fact that {\small{$S_{0} + 1/2 \cdot S_{\ferm}$}} has total ghost degree 0, allows to conclude that necessarily, contacting {\small{$S_{\ferm}$}} with a term $f$ in \scM \ we could only find a term in {\small{$\M \otimes \B_{1}$}}. Finally, similarly to what seen for the proof of the graded coassociativity of the coproduct {\small{$\Delta$}}, also the graded compatibility of the coaction {\small{$\omega_{R}$}} and the coproduct {\small{$\Delta$}} follows directly from the properties of the Poisson structure as well as from the fact that \swS \ is a solution to the classical master equation. This allows to conclude that {\small{$(\M, \omega_{R})$}} defines a degree-1 right comodule over \scB. 
\end{proof}

\subsection{The isomorphism of cohomology complexes}
Having described in full generality how to associate to any BV spectral triple the induced pair of coalgebra and comodule, that is, knowing how to perform the step:
$$(\mathcal{A}_{\BV}, \mathcal{H}_{\BV}, D_{\BV}, J_{\BV}) \longrightarrow (\mathcal{B}, \mathcal{M})$$
we can finally state and proof the main result of this section: the BV cohomology induced by a BV-extended theory \swXS \ is isomorphic to the Hochschild cohomology determined by the pair {\small{$(\mathcal{B}, \mathcal{M})$}}, as described in Definition \ref{def: Hochschild cohom, graded}.

\begin{theorem}
\label{Theorem: BV and Hochschild}
Let \BVsp\ be the BV spectral triple associated to an initial spectral triple \isp\ and corresponding to a BV-extended theory \swXS. Then, given {\small{$(\mathcal{B}, \mathcal{M})$}} the associated pair of a 1-shifted graded coalgebra {\small{$\mathcal{B}$}} together with the degree-1 right comodule {\small{$\mathcal{M}$}} over it, the following isomorphism holds at the level of the induced cohomology complexes:
$$\mathcal{C}^{\bullet}_{\BV}(\widetilde{X}, d_{\widetilde{S}}) \cong \mathcal{C}^{\bullet}_{\Hoch, \Delta}(\mathcal{B}, \mathcal{M}) $$
\end{theorem}

\begin{proof}
To prove the statement, we need to define a collection {\small{$\{\Phi_{k}\}_{k \in \mathbb{Z}}$}} of maps $$\Phi^{k}: \mathcal{C}_{\BVt}^{k}(\widetilde{X}, d_{\widetilde{S}}) \longrightarrow \mathcal{C}_{\Hoch, \Delta}^{k}(\mathcal{B},  \mathcal{M})$$
between the degree $k$ cochain spaces, which is compatible with the two corresponding coboundary operators, that is, which satisfies the condition
$$d_{H, \Delta} \circ  \Phi^{k}|_{\varphi} = \Phi^{k+1} \circ d^{k}_{\widetilde{S}}|_{\varphi}$$
for all $\varphi \in \mathcal{C}_{\BVt}^{k}(\widetilde{X}, d_{\widetilde{S}})$. As briefly recalled in Section \ref{Subsect: The BV cohomology complex}, any BV-extended theory \swXS\ naturally induces a cohomology theory, called the {\emph{\cBRST\ cohomology}}, whose cochain spaces are:
$$\mathcal{C}_{\BVt}^{k}(\widetilde{X}, d_{\widetilde{S}}) := [Sym_{\mathcal{O}_{X_{0}}}(\widetilde{X})]^{k} = [\mathcal{O}_{\widetilde{X}}]^{k},  $$
for {\small{$ k \in \mathbb{Z}$}}, with coboundary operator given by
$$d_{\widetilde{S}}: \mathcal{C}_{\BVt}^{\bullet}(\widetilde{X}, d_{\widetilde{S}}) \rightarrow \mathcal{C}_{\BVt}^{\bullet +1}(\widetilde{X}, d_{\widetilde{S}}), \quad \quad \mbox{ for } \quad \quad d_{\widetilde{S}}:= \{ \widetilde{S}, - \}.$$ 
\noindent
Hence, from the definitions of \scB\ and \scM \ given respectively in \eqref{Eq: def B} and \eqref{eq: def M}, it immediately follows that
$$\mathcal{C}^{k}_{\Hoch, \Delta}(\mathcal{B}, \mathcal{M}) := \M \otimes T^{k}(\B) \cong \Sym_{\mathcal{M}}^{k}(\mathcal{B}) \cong \Sym_{\mathcal{O}_{X_{0}}}^{k}(\widetilde{X}) =: \mathcal{C}^{k}_{BV}(\widetilde{X}, d_{\widetilde{S}}),$$
in any degree {\small{$k \in \mathbb{Z}$}}. This isomorphism becomes evident if, given an element {\footnotesize{$\varphi$}} in {\footnotesize{$\mathcal{C}^{k}_{BV}(\widetilde{X}, d_{\widetilde{S}})$}} we describe it as {\footnotesize{$\varphi = \sum_{I = i_{1}, \dots, i_{p}} f_{I} D^{i_{1}}_{1} \dots D^{i_{p}}_{p}$}} where $f_{I}$ is an element in {\footnotesize{$\mathcal{O}_{X_{0}}$}}, {\footnotesize{$D^{i_{j}}_{j}$}} generically denotes any ghost/antifield/anti-ghost field with ghost degree {\small{$i_{j}$}} and the indices {$I$} have variable cardinality but they satisfy the condition {\small{$i_{1} + \dots + i_{p} = k$}}. Then, fixed {\small{$k$}} in {\small{$\mathbb{Z}$}}, we explicitly define the isomorphism map {\small{$\Phi_{k}$}} to be:
$$ \Phi^{k}: \mathcal{C}^{k}_{BV}(\widetilde{X}, d_{\widetilde{S}}) \rightarrow \mathcal{C}^{k}_{\Hoch, \Delta}(\mathcal{B}, \mathcal{M}) \quad \quad \sum_{I} f_{I} D^{i_{1}}_{1} \dots D^{i_{p}}_{p} \to \sum_{I} f_{I} \otimes y_{1}^{i_{1}} \otimes \dots \otimes y_{p}^{i_{p}}   $$
To conclude the proof we just need to check the compatibility of this isomorphism on the spaces of cochains with the two coboundary operators {\small{$d_{\widetilde{S}}$}} and {\small{$d_{H, \Delta}$}}. Due to the linearity of the maps {\small{$\Phi^{\bullet}$}} as well as of the coboundary operators, we can restrict ourselves to consider an element {\small{$\varphi$}} given by one single summand. Moreover, as both {\small{$d_{\widetilde{S}}$}} and {\small{$d_{H, \Delta}$}} act as graded derivation, the compatibility of the isomorphism with the two coboundary operators will immediately follows from verifying this holds on the generators:
$$\Phi^{1}(\{ \widetilde{S}, f\}) = \omega_{R}(\Phi^{0}(f)) \quad \quad \Phi^{i_{j}+1}(\{ \widetilde{S}, D^{i_{j}}_{j}\}) = \Delta(\Phi^{i_{j}}(D^{i_{j}}_{j})).$$
As both these equation follows straightforwardly from the definitions of {\small{$\omega_{R}$}} and {\small{$\Delta$}}, we can then conclude that the collection of maps {\small{$\{ \Phi^{k}\}_{k \in \mathbb{Z}}$}} is a cochain isomorphism for the cochain complexes {\small{$\mathcal{C}^{\bullet}_{\BV}(\widetilde{X}, d_{\widetilde{S}})$}} and {\small{$ \mathcal{C}^{\bullet}_{\Hoch, \Delta}(\mathcal{B}, \mathcal{M})$}}
\end{proof}

\noindent
{\emph{Note:}} The theorem was proved for a generic BV spectral triple \BVsp, not only in the specific case of {\small{$\mathrm{U}(n)$}}-gauge theories we focused our attention on in Section \ref{Sect: BV spectral triple}. This implies that the Hochschild cohomology complex associated to a BV spectral triple and the BV cohomology of the corresponding BV extended theory \swXS \ are isomorphic also in case one considers reducible gauge theories with level of reducibility $L>0$ and not only in the irreducible case, which involves only ghost fields of ghost degree $1$. The only additional condition one has to take into account is the need for the BV action {\small{$S_{\BV}$}} to be linear in the antifields/anti-ghost fields and quadratic in the ghost fields (cf. Lemma \ref{lemma: B is a graded coalgebra}). This is a property that comes naturally in the context of BV spectral triple associted to {\small{$\mathrm{U}(n)$}}-gauge theories. However, one could remark that the linearity in the  antifields/anti-ghost fields is not a very restrictive condition as it is satisfied each time one considers a gauge theory whose algebra of gauge symmetries is closed on-shell.\\
\\
In this section we reached the goal of proving, for the first time in the framework provided by NCG and spectral triples, that the \cBRST\ complex coincides with the Hochschild cohomology complex of a suitable coalgebra \scB\ over a comodule \scM. We expect that this new point of view on the \cBRST\ cohomology could not only contribute to its explicit computation, by relating it to another already well-known and widely investigated cohomology theory, but also to provide the BV complex with a more geometrical understanding within the framework of NCG.\\

\section{Auxiliary fields and the total spectral triple}
\label{Sect: gauge-fixing procedure in terms of NCG}
\noindent
As explained and motivated in Section \ref{Subsec: Auxiliary fields and the gauge-fixing process}, determining the BV-extended pair \swXS \ and its induced BV complex might not be the final result we aim for: indeed, the presence of antifields and anti-ghost fields, both in the configuration space \swX \ as well as in the extended action \swS, imposes the implementation of a gauge-fixing procedure before being able to perform some computations for the theory, such as determining the elements of the S-matrix. However, a preliminary step is required before being able to apply such a gauge-fixing, that is, the introduction of auxiliary fields (cf. Definition \ref{auxiliary pair}). This is just a technical step: the minimal collection of auxiliary pairs to introduce, their type and their number, only depends on the {\emph{level of reducibility}} of the theory (cf. Definition \ref{def: level of reducibility}), as stated in the following theorem, where {\small{$\Psi$}} denotes a gauge-fixing fermion as defined in Definition \ref{def gauge fixing fermion}.

\begin{theorem}
\label{teorema auxiliary fields caso generale}
Let {\small{$(\widetilde{X}, \widetilde{S})$}} be a BV-extended theory with level of reducibility L. Then, to determine a total theory {\small{$(X_{\t}, S_{\t})$}} whose gauge-fixed action {\small{$S_{\t}|_{\Psi}$}} is a proper solution to the classical master equation, a collection of auxiliary pairs {\small{$\{ (B_{i}^{j}, h_{i}^{j}) \}$}}, with {\small{$i= 0, \dots, L$}} and {\small{$j=1,\dots, i+1$}}, has to be introduced, which is completely determined by imposing that 
$$\deg(B_{i}^{j}) = j -i -2  \quad \mbox{ if j is odd}, \quad \quad \deg(B_{i}^{j}) = i -j +1 \quad \mbox{ if j is even}.$$
\end{theorem}

\noindent
This theorem was first proved inductively by Batalin and Vilkovisky (cf. \cite{BV1}, \cite{BV2}).\\
\\
In order for the language of noncommutative geometry to provide a suitable and worthwhile point of view to approach the study of the BV construction, we have to require that also this step of the process fits within the spectral triple framework. To reach this goal we are going to introduce the concept of {\emph{total spectral triple}}.\\
\\
{\bf The total spectral triple.} \ Let us start by briefly recalling that, given a BV-extended theory \swXS, the introduction of auxiliary fields determines the creation of a so-called total theory \stXS\ with:
$$X_{\t} := \widetilde{X} + \big\{ \mbox{auxiliary fields and antifields} \big\} \qquad \mbox{and} \qquad S_{\t}:= \widetilde{S} + S_{\aux}.$$
Our goal is to introduce a new object, called a \textquotedblleft total spectral triple\textquotedblright, able to encode the total theory \stXS. Precisely, we want to determine how to perform the step: 
$$\begin{array}{ccc}
(\mathcal{A}_{\BV}, \mathcal{H}_{\BV}, D_{\BV}, J_{\BV})& ---------\rightarrow  &(\mathcal{A}_{\t}, \mathcal{H}_{\t}, D_{\t}, J_{\t})\\
\mbox{\small{BV spectral triple}} & \mbox{\tiny{+ auxiliary fields/antifields}} & \mbox{\small{total spectral triple}}
\end{array}
$$
that is, to explain how to extract all the needed information from a given BV spectral triple in order to construct the associated total spectral triple. Moreover, we also need that this further passage coherently enters the framework we are revealing, step by step, for performing the BV construction in the language of finite spectral triple. Up to now indeed, we demonstrate how to define a BV spectral triple such that not only it encodes the BV-extended theory \swXS\ but also it determines a pair {\small{$(\mathcal{B}, \mathcal{M})$}} of a coalgebra and a comodule on it such that their induced Hochschild complex coincides with the BV complex defined by \swXS. To have that our translation of the BV construction in terms of NCG is a coherent and consistent process, the total spectral triple has to present the same properties. Precisely: \vspace{1mm}
\noindent
\begin{enumerate}
\item not only  the total spectral triple has to encode the total theory but the way how, given a total spectral triple, one can reconstruct the corresponding total theory should be exactly the same method used to extract from a BV spectral triple the corresponding extended theory \swXS. \vspace{2mm}

\item By implementing on \TOTsp\ the same construction described in Section \ref{Sect: BV and Hochschild complex} for the BV spectral triple, we can once again determine a pair {\small{$(\mathcal{B}_{\t}, \mathcal{M}_{\t})$}} with {\small{$\mathcal{B}_{\t}$}} a coalgebra and {\small{$\mathcal{M}_{\t}$}} a comodule over it. To be able to assert that this whole formalism is indeed consistent, a necessary condition to fulfill is that the Hochschild complex of {\small{$(\mathcal{B}_{\t}, \mathcal{M}_{\t})$}} is isomorphic to the BV complex defined by the total theory \stXS. \vspace{2mm}
\end{enumerate}
\noindent
All this can be summarized in the diagram appearing in Figure \ref{diag. tot sp triple}: our aim is to determine a spectral triple \TOTsp\ that could make the diagram commutative.

\begin{figure}[h]
\begin{picture}(400, 90)
\put(12, 80){\small{$(\mathcal{A}_{0}, \mathcal{H}_{0}, D_{0})$}}
\put(68, 85){\tiny{+ ghost fields}}
\put(65, 83){\vector(1,0){52}}
\put(42, 66){\tiny{(bosonic)}}
\put(40, 60){\tiny{induced th.}}
\put(36, 75){\vector(0,-1){22}}
\put(20, 40){\small{$(X_{0}, S_{0})$}}
\put(63, 45){\tiny{+ ghost fields}}
\put(55, 43){\vector(1,0){65}}

\put(120, 80){\small{$(\mathcal{A}_{\BV}, \mathcal{H}_{\BV}, D_{\BV}, J_{\BV})$}}
\put(217, 85){\tiny{+ auxiliary flds}}
\put(197, 45){\tiny{+ auxiliary fields}}
\put(215, 83){\vector(1,0){57}}
\put(143, 66){\tiny{(fermionic)}}
\put(143, 60){\tiny{induced th.}}
\put(139, 75){\vector(0,-1){22}}
\put(125, 40){\small{$(\widetilde{X}, \widetilde{S})$}}
\put(155, 43){\vector(1,0){107}}
\put(139, 35){\vector(0,-1){24}}
\put(193, 75){\vector(0,-1){42}}
\put(180, 27){\small{$(\mathcal{B},  \mathcal{M})$}}
\put(193, 24){\vector(0,-1){14}}
\put(125, 0){\small{$\mathcal{C}^{\bullet}_{\BV}(\widetilde{X}, d_{\widetilde{S}})\cong\mathcal{C}^{\bullet}_{\Hoch, \Delta}(\mathcal{B}, \mathcal{M})$}}

\put(275, 80){\small{$(\mathcal{A}_{\t}, \mathcal{H}_{\t}, D_{\t}, J_{\t})$}}
\put(284, 66){\tiny{(fermionic)}}
\put(284, 60){\tiny{induced th.}}
\put(280, 75){\vector(0,-1){22}}
\put(265, 40){\small{$(X_{\t},  S_{\t})$}}
\put(280, 35){\vector(0,-1){27}}
\put(330, 75){\vector(0,-1){42}}
\put(315, 27){\small{$(\mathcal{B}_{\t},  \mathcal{M}_{\t})$}}
\put(330, 24){\vector(0,-1){14}}
\put(255, 0){\small{$\mathcal{C}^{\bullet}_{\BV}(X_{\t}, d_{S_{\t}})\cong\mathcal{C}^{\bullet}_{\Hoch, \Delta}(\mathcal{B}_{\t}, \mathcal{M}_{\t})$}}
\end{picture}
\caption{}
\label{diag. tot sp triple}
\end{figure}

\noindent
Before explaining in details the construction, let us state the precise definition of what will be for us a {\emph{total spectral triple}}.

\begin{definition}
Let \isp\ be a finite spectral triple with induced gauge theory \siXS. Then an (eventually mixed KO-dimension) finite spectral triple \TOTsp\ is a {\emph{total spectral triple}} for \isp\ if the pair: 
$$X_{\t}:= \big(\mathcal{Q}_{f}^{*}[1] + X^{*}_{0}[1] \big) \oplus \big(X_{0} + \mathcal{Q}_{f} \big) \oplus \big( \mathcal{R}_{f}^{*}[1] \oplus \mathcal{R}_{f} \big), $$
where 
$$\mathcal{H}_{\t, f} = (\mathcal{Q}_{f}^{*}[1] \oplus \mathcal{Q}_{f}) \oplus (\mathcal{R}_{f}^{*}[1] \oplus \mathcal{R}_{f} ) \quad \mbox{and} \quad X_{0} \cong [\Omega^{1}(\mathcal{A}_{0})]_{\text{s.a.}}$$
and 
$$S_{\t}:= S_{0} + \frac{1}{4} \langle J_{\t}(\varphi), D_{\t} \varphi \rangle\ \in [\mathcal{O}_{X_{\t}}]^{0}$$
for $\varphi$ a generic element in {\small{$\mathcal{H}_{\t, f}$}}, defines the total theory {\small{$(X_{\t}, S_{\t})$}} associate to the initial gauge theory \siXS.
\end{definition}

\noindent
{\emph{Note:}} we observe that, by definition, we are in particular imposing that any total spectral triple has to satisfy condition $(1)$. Indeed, this definition reproduces exactly the same structure already observed for the BV spectral triple: 
\begin{enumerate}[$\blacktriangleright$]
\item we expect the total configuration space \stX\ to be encoded in the total spectral triple as the total Hilbert space \HTOT; \vspace{2mm}
\item the total action \stS\ is asked to coincide with the fermionic action induced by the operator \DTOT\ and the real structure \JTOT.
\end{enumerate}\vspace{1mm}
\noindent
Hence, our starting point is the BV spectral triple constructed in Section \ref{Sect: BV spectral triple} and given by:
$$(\mathcal{A}_{\BV}, \mathcal{H}_{\BV}, D_{\BV}, J_{\BV}) = (M_{n}(\mathbb{C}), \mathcal{Q}^{*}[1] \oplus \mathcal{Q}, D_{\BV}, J_{\BV}))$$
where 
$$\mathcal{Q} = [M_{n}(\mathbb{C})]_{0} \oplus [M_{n}(\mathbb{C})]_{1}$$
with effective BV Hilbert space 
$$\mathcal{H}_{\BV, f} =  \mathcal{Q}_{f}^{*}[1] \oplus \mathcal{Q}_{f} \quad \mbox{for} \quad \mathcal{Q}_{f} = [i\mathfrak{su}(n)]_{0} \oplus [i\mathfrak{su}(n)]_{1}. $$
We now proceed to the definition of the different parts of \TOTsp, starting with its algebra \ATOT.
\\
\\
\noindent
{\bf The algebra $\mathcal{A}_{\t}$.} Similarly to what done for the construction of a BV algebra \ABV, also in this case the algebra \ATOT\
 will be defined only a posteriori as the maximal unital algebra completing the triple {\small{$(\mathcal{H}_{\t}, D_{\t}, J_{\t})$}}, which we are soon going to define, to a spectral triple with KO-dimension $1$ (cf. Theorem \ref{Theor: total spectral triple}). However, making a comparison with how the extension of an initial spectral triple to a BV spectral triple appears at the level of the corresponding algebra, we could make a first guess about how the further extension via the introduction of auxiliary fields will act at the algebra level. Indeed, let us notice that the algebra {\small{$\mathcal{A}_{0}$}} incorporates the information about the initial configuration space, and hence about the initial and physically relevant part of the theory, as it holds that {\small{$X_{0}:=[\Omega^{1}(\mathcal{A}_{0})]_{\text{s.a.}}$}}. Because the underlying (existing) physical part did not change when we introduced the ghost/anti-ghost fields, this resulted in the algebra \ABV\ being equal to \AI. Hence, as introducing the auxiliary fields once again does not interfere with the existing physical content of the theory, it is not surprising that the algebra \ATOT\ will coincide with \AI. By using the tool of the Krajewski diagrams or by a direct computation one can verify that, what suggested by the physical intuition is actually mathematically true, that is:
$$\mathcal{A}_{\t}= \mathcal{A}_{\BV} = \mathcal{A}_{0} = M_{n}(\mathbb{C}).$$
\\
\noindent
{\bf The Hilbert space $\mathcal{H}_{\t}$.} In order to construct \HTOT\ one has first to determine how many and which kind of auxiliary pairs are required by the theory. Looking at the graded structure of {\small{$\mathcal{H}_{\BV}$}}, one immediately computes the level of reducibility of the theory, which in our case is {\small{$L=0$}}. Hence, according to Theorem \ref{teorema auxiliary fields caso generale}, the presence of {\small{$n^{2}-1$}} ghost fields {\small{$C_{q}$}} of degree $1$, which are encoded in the term {\small{$[i\mathfrak{su}(n)]_{1} \subset \mathcal{H}_{\BV, f}$}}, determines the appearance of the summands
 $$[\mathfrak{su}(n)]_{-1} \oplus [i\mathfrak{su}(n)]_{0} \subseteq [M_{n}(\mathbb{C})]_{-1} \oplus [M_{n}(\mathbb{C})]_{0}$$
in {\small{$\mathcal{H}_{\t, f} \subseteq\mathcal{H}_{\t}$}}, which corresponds to {\small{$n^{2}-1$}} trivial pairs  
 $$\{(B_{q}, h_{q})\}_{q= 1, \dots, n^{2}-1}\quad  \mbox{with} \quad \deg(B_{q})=  -1,  \quad \mbox{and} \quad
\deg(h_{q})= 0.$$
 \noindent
Recalling that the BV formalism forces the introduction of all the antifields corresponding to the auxiliary fields just listed, we conclude that the total Hilbert space is given by: 

\begin{equation}
\label{eq: HTOT}
\mathcal{H}_{\t}: = \mathcal{H}_{\BV} \oplus \mathcal{H}_{\aux}, \quad \mbox{where} \quad \mathcal{H}_{\BV} = [M_{n}(\mathbb{C})]^{\oplus 4}, \quad \mathcal{H}_{\aux} = [M_{n}(\mathbb{C})]^{\oplus 4} 
\end{equation}
\noindent
and the effective part of {\small{$\mathcal{H}_{\aux}$}}, which determines the domain of the induced fermionic action, is
\begin{equation}
\label{eq: Rf}
\mathcal{H}_{\aux, f} = \mathcal{R}_{f}^{*}[1] \oplus \mathcal{R}_{f} \quad \quad \mbox{for} \quad \quad 
\mathcal{R}_{f}:= [\mathfrak{su}(n)]_{-1} \oplus [i \mathfrak{su}(n) ]_{0}.
\end{equation}

\noindent 
For clarity in the notation, when realizing \HTOT\ as direct sum of matrix algebras, each of them carrying a degree, given the decomposition of \HTOT\ as 
$$\mathcal{H}_{\t} = (\mathcal{Q}^{*}[1] \oplus \mathcal{Q}) \oplus (\mathcal{R}^{*}[1] \oplus \mathcal{R}) $$
the matrix algebras contributing to each summand will be ordered by increasing degree. Explicitly, the elements in {\small{$\mathcal{H}_{\aux}$}} can be described as vectors of the type: {\small{$[h^{*}_{i}, B_{i}, B^{*}_{i}, h_{i}]$}}. This convention will be used in the following paragraph, when describing the operator \DTOT\ explicitly as a matrix.

\begin{oss}
Even though at this point the fact that the auxiliary fields {\small{$B_{i}$}} should be embedded in {\small{$M_{n}(\mathbb{C})$}} as the term {\small{$[\mathfrak{su}(n)]$}} while the fields {\small{$h_{i}$}} as {\small{$[i\mathfrak{su}(n)]$}} might appear arbitrary, one can check that this choice is the only one which would allow to comply with the requirement of having an induced auxiliary action {\small{$S_{\aux}$}} with {\emph{real}} coefficients. 
\end{oss}

\noindent
{\bf The operator $D_{\t}$.} What determines the explicit form of the operator \DTOT\ is the condition that the induced fermionic action {\small{$S_{\ferm}$}} has to coincide with the total action \stS. Precisely, it should hold that:
$$\frac{1}{2} S_{\ferm}(\varphi) := \frac{1}{4} \langle J_{\t}(\varphi), D_{\t}\varphi \rangle \cong S_{\BV} + S_{\aux} \quad \mbox{for} \quad S_{\aux} := \sum_{i=1}^{n}B^{*}_{i} h_{i},$$
where we are implicitly identifying a linear functional {\small{$S_{\ferm}$}} formally defined on {\small{$\mathcal{H}_{\t, f}$}} with its representation as polynomial in {\small{$[\mathcal{O}_{X_{\t}}]^{0}$}}. Because in the action \stS\ there are no terms involving both the auxiliary fields as well as the ghost fields/anti-ghost fields already introduced in \swX, the operator {\small{$D_{\t}$}} has to have the structure of a block matrix with
\begin{equation}
\label{Eq: DTOT}
D_{\t} = \begin{pmatrix} D_{\BV} & 0 \\ 0 & D_{\aux} \end{pmatrix}
\end{equation}
for \DBV\ and {\small{$D_{\aux}$}} two self-adjoint operators. While \DBV\ is the operator appearing in the BV spectral triple, acting and taking values on the summand \HBV\ in {\small{$\mathcal{H}_{\t}$}} and contributing the term \SBV\ to the fermionic action, {\small{$D_{\aux}$}} can be seen as an operator {\small{$D_{\aux}: \mathcal{H}_{\aux} \rightarrow  \mathcal{H}_{\aux}$}}, which will contribute the term {\small{$S_{\aux}$}}. Moreover, as {\small{$S_{\aux}$}} contains only terms which are linear in the auxiliary fields as well as in their antifields, the operator {\small{$D_{\aux}$}}, seen as a block matrix, necessarily is of the form 
$$D_{\aux}:=  \begin{pmatrix}0 & T \\
T^{*} & 0 
\end{pmatrix}
$$
for {\small{$T: \mathcal{R}\rightarrow \mathcal{R}^{*}[1]$}} a linear operator with constant entries. To explicitly discover the matrix representation of this operator $T$ one could start by imposing the condition that the induced action has to be of ghost degree $0$. Even though this requirement forces several entries of the matrix to be zero, still this is not enough to completely determine the matrix as some of the auxiliary fields have the same ghost degree. Actually, what need to be implemented in $T$ is the information about the pairing existing between the auxiliary fields: precisely, the only non-zero entries are in correspondence of auxiliary fields which belong to the same auxiliary pair. For our model, we have that the matrix $T$ has the following explicit expression: \vspace{1mm}
{\small{
\begin{gather*}  
T:= \begin{pmatrix} 
0 & 0 \vspace{1mm}\\
0 & 2\cdot \mathrm{Id}\vspace{1mm}\\
\end{pmatrix}
\end{gather*}
}}

\noindent
where each entry {\small{$2 \cdot \mathrm{Id}$}} contributes the summand {\small{$\sum_{i = 1}^{n^{2}-1} B^{*}_{i}h_{i}$}} in the auxiliary action {\small{$S_{\aux}$}}. We remark how the explicit description of the operator $T$ clearly reflects the cohomological triviality of the extension of the theory via the introduction of the auxiliary fields.\\
\\
\noindent
{\bf The real structure $J_{t}$.} Finally, to complete the description of the total spectral triple we still have to specify its real structure. Because we are simply enlarging the BV spectral triple to include additional fields, it is natural to consider as real structure the one obtained by extending the domain of {\small{$J_{\BV}$}} from \HBV \ to \HTOT:
\begin{equation}
\label{Eq: JTOT}
J_{\t}:  \mathcal{H}_{\t} \rightarrow \mathcal{H}_{\t}, \quad \mbox{with} \quad J_{\t}(\varphi_{j}): = i \cdot (\phi_{j}^{*}) 
\end{equation}
for $\phi_{j} \in  [M_{n}(\C)]_{j}$ an element in the direct sum defining \HTOT \ and $*$ denoting the matrix adjoint.\\
\\
To conclude, we summarize and formalize this construction of the total spectral triple for our {\small{$\mathrm{U}(n)$}}-model in the following theorem.

\begin{theorem}
\label{Theor: total spectral triple}
Given a finite spectral triple {\small{$(\mathcal{A}_{0}, \mathcal{H}_{0}, D_{0}) = (M_{n}(\mathbb{C}), \mathbb{C}^{n}, D_{0})$}} inducing a {\small{$\mathrm{U}(n)$}}-gauge theory, let 
$$(M_{n}(\mathbb{C}), [M_{n}(\mathbb{C})]_{-2} \oplus [M_{n}(\mathbb{C})]_{-1} \oplus[M_{n}(\mathbb{C})]_{0} \oplus [M_{n}(\mathbb{C})]_{1}, D_{\BV}, J_{\BV})$$ 
be the associated BV spectral triple, with effective BV Hilbert space {\small{$\mathcal{H}_{\BV, f} = \mathcal{Q}^{*}_{f}[1] \oplus \mathcal{Q}_{f}$}}. Then, the corresponding total spectral triple is
$$(\mathcal{A}_{\t}, \mathcal{H}_{\t}, D_{\t}, J_{\t}) = (M_{n}(\mathbb{C}), \mathcal{H}_{\BV}\oplus \big[ [M_{n}(\mathbb{C})]^{\oplus 2}_{-1} \oplus [M_{n}(\mathbb{C})]^{\oplus 2}_{0} \big], D_{\t}, J_{\t}),$$
for {\small{$\mathcal{H}_{\t, f} := \mathcal{H}_{\BV, f} \oplus (\mathcal{R}^{*}_{f}[1] \oplus \mathcal{R}_{f})$}}, where {\small{$\mathcal{R}_{f}$}} is given in \eqref{eq: Rf}, {\small{$D_{\t}$}} is as described in \eqref{Eq: DTOT} and \JTOT\ is defined in \eqref{Eq: JTOT}. Moreover, the choice of {\small{$\mathcal{A}_{\t}= M_{n}(\mathbb{C})$}} is maximal.
\end{theorem}

\noindent
{\emph{Note:}} As already done for the BV spectral triple, also in this case all the variables/components in {\small{$\mathcal{H}_{t, f}$}} and \DTOT\ have to be treated as graded variables, where their real or Grassmannian parity coincides with the even/odd parity of their ghost degree.
\begin{proof}
Checking that {\small{$(\mathcal{A}_{\t}, \mathcal{H}_{\t}, D_{\t}, J_{\t})$}} defines a spectral triple is pretty straightforward: indeed, from the representation of \DTOT\ as matrix, it is clearly a self-adjoint operator. Moreover, because the map \JTOT\ has been simply obtained as extension of the real structure \JBV, it clearly defines a real structure on the whole Hilbert space \HTOT. Also, by a direct computation, one can verify that the operator {\small{$D_{\aux}$}} commutes with the real structure: 
$$\begin{array}{ccc} D_{\aux}J_{\aux} = J_{\aux}D_{\aux} \\
\mbox{on the bosons} 
\end{array}
$$
determining a real spectral triple with mixed KO-dimension. For what concerns the conditions about the algebra \ATOT\ and its maximality, one can verify this part of the statement either via a direct computation or by constructing the corresponding Krajewski diagram. Finally, to assert that the spectral triple \TOTsp\ just defined is in particular a total spectral triple for the gauge theory \siXS\ we have to verify that:
$$X_{\t} \cong \mathcal{H}_{\t, f} + (X_{0} \oplus X_{0}^{*}[1])\qquad \mbox {and} \qquad S_{\t}  \cong S_{0} + \frac{1}{4} \langle J_{\t}(\varphi), D_{\t} \varphi \rangle$$
for {\small{$\varphi$}} a generic element in {\small{$\mathcal{H}_{\t, f}$}}. However, both these identities follows immediately from the explicit definitions given of  {\small{$\mathcal{H}_{\t}$}} and \DTOT\ in \eqref{eq: HTOT} and  \eqref{Eq: DTOT} respectively.
\end{proof}

\begin{oss}
Up to now, following the line indicated by the BV construction, we have shown how to determine, given a BV spectral triple, the associated total spectral triple. However, one can also follow the opposite direction and, given a total spectral triple \TOTsp, being able to eliminate the auxiliary fields/contractible pairs to recover the underline BV spectral triple. Indeed, to do so one has only to be able to identify the auxiliary fields and separate them from all the ghost/anti-ghost fields. To implement this selection among the fields in \HTOT\ one has simply to look at the structure of the operator \DTOT: the part of \DTOT\ containing just constants is the part of it acting on the subspace {\small{$\mathcal{H}_{\aux}$}} in \HTOT\ while the other components of \DTOT, depending on ghost/anti-ghost fields, form the part of it acting on \HBV. Hence, looking at the structure of \DTOT, we recover not only \DBV\ and {\small{$D_{\aux}$}} but also the splitting of \HTOT\ as {\small{$\mathcal{H}_{\t} = \mathcal{H}_{\BV} \oplus \mathcal{H}_{\aux}$}}. Therefore, having already determined \HBV\ and \DBV, one has only to define {\small{$J_{\BV}$}} as restriction of {\small{$J_{\t}$}} to \HBV\ and to fix {\small{$\mathcal{A}_{\BV}:=\mathcal{A}_{\t}$}}. This would complete the reconstruction of the underlying BV spectral triple.
\end{oss}

To be able to assert that the total spectral triple just defined for our model acts exactly as indicated by conditions $(1)$ and $(2)$ we still have to check that \TOTsp\ verifies condition $(2)$, as the requirement in $(1)$ has been directly implemented in the definition itself of {\emph{total spectral triple.}} Explicitly, we have to verify that, by applying to \TOTsp\ exactly the same procedure explained in Section \ref{Sect: BV and Hochschild complex}, the resulting pair {\small{$(\mathcal{B}_{\t}, \mathcal{M}_{\t})$}} induces a Hochschild complex such that:
$$(\mathcal{C}^{\bullet}_{H}(\mathcal{B}_{\t}, \mathcal{M}_{\t}), d_{H}) \cong (\mathcal{C}^{\bullet}_{\BVt}(X_{\t}, d_{S_{t}}), d_{S_{t}}),$$
for {\small{$\mathcal{C}^{\bullet}_{\BVt}(X_{\t}, d_{S_{t}})$}} the effective BV complex defined by the pair \stXS. \\
\\
\noindent
{\bf The coalgebra $\B_{\t}$} By implementing the formula in \eqref{Eq: def B} starting from the total spectral triple one finds that the homogeneous components of {\small{$\B_{\t}$}} \ are 
$$\B_{\t, 0} = [\mathcal{O}_{X_{0}}]^{\leq(\deg(S_{0})-1)} \oplus [\mathfrak{su}(n)] \oplus [i \mathfrak{su}(n)], \quad \B_{\t,-1} = X_{0}^{*} [1] \oplus [\mathfrak{su}(n)] \oplus [i \mathfrak{su}(n)],$$
and $$\B_{\t,m} = [\mathcal{Q}_{f}]^{m}, \quad \B_{\t,-m} = [\mathcal{Q}_{f}^{*}[1]]^{-m}$$ for {\small{$m>1$}}, where additional terms appear only in the homogeneous components of degree $-1$ and $0$, due to the presence of the auxiliary fields {\small{$(B_{p}, h_{p})$}} and their corresponding antifields {\small{$(B^{*}_{p}, h^{*}_{p})$}}. Concerning its coproduct structure, given a homogeneous element {\small{$y^{a}$}}, we define 
\begin{equation}
\Delta_{t}(y^{a}):= \left\lbrace S_{0} + \frac{1}{2} S_{\ferm, t}, \  y^{a} \right\rbrace     
\end{equation}
where {\small{$S_{\ferm, t}$}} should be seen as the fermionic action induced by the total spectral triple.
However, the splitting  of \HTOT\ in the direct sum {\small{$\mathcal{H}_{\BV} \oplus \mathcal{H}_{\aux}$}} and the diagonal structure of \DTOT\ allows for a more explicit and direct description of {\small{$\mathcal{B}_{\t}$}} as extension of the algebra {\small{$\mathcal{B}$}} and the coproduct structure on {\small{$\mathcal{B}_{\t}$}} can be written as
\begin{equation}
\Delta_{t}(y^{a}):= \left\lbrace S_{0} + S_{\BV}, \  y^{a} \right\rbrace  \quad \mbox{or} \quad \Delta_{t}(y^{a}):= \left\lbrace S_{\aux}, \  y^{a} \right\rbrace \end{equation}
depending if {\small{$y^{a} \in \mathcal{H}_{\BV}$}} or 
{\small{$y^{a} \in \mathcal{H}_{\aux}$}}, respectively.
\\
\\
\noindent 
{\textbf{The bimodule {\small{$\mathcal{M}_{\t}$}}}} Accordingly to what appearing in Equation \eqref{eq: def M}, the comodule {\small{$\mathcal{M}_{\t}$}} defined to be {\small{$\mathcal{M}_{\t}:= \langle \Omega^{1}(\mathcal{A}_{\t})_{\text{s.a.}}\rangle = \M \cong \mathcal{O}_{X_{0}} $}}. The fact that {\small{$\mathcal{M}_{\t}$}} coincides with \scM \ is a direct consequence of the fact that {\small{$\A_{\t} = \A_{\BV}$}}. Finally, the right coaction {\small{$\omega_{R, t}$}}of {\small{$\mathcal{M}_{\t}$}} over {\small{$\mathcal{B}_{\t}$}} is given by

\begin{equation}
\label{Eq: def coaction omega t}
    \omega_{R, t} (f):=\left\lbrace S_{0} + \frac{1}{2} S_{\ferm, t}, f\right\rbrace = \left\lbrace S_{0} + S_{\BV}, f\right\rbrace = \omega_{R}(f). 
\end{equation}

The fact that the coaction {\small{$\omega_{R, t}$}} coincides with the coaction {\small{$\omega_{R}$}} follows from the fact that the only non-zero contribution to the expression in \ref{Eq: def coaction omega t} comes from the contraction of f with an antifield {\small{$x^{*}_{i}$}}. As no antifields appear in {\small{$S_{\aux}$}}, neither the comodule not its coaction change in the passage from the BV spectral triple to its extension via auxiliary fields: {\small{$(\M_{t}, \Delta_{t}) \cong (\M, \Delta)$}}.\\
\\
\noindent
To conclude the section, given this pair {\small{$(\mathcal{B}_{\t}, \mathcal{M}_{\t})$}}, we should verify that its induced Hochschild complex is isomorphic to the BV complex of \stXS. Proving this claim is exactly the purpose of the following theorem. 

\begin{theorem}
\label{Theo: BV tot}
Let {\small{$(\mathcal{A}_{\t}, \mathcal{H}_{\t}, D_{\t}, J_{\t})$}} be a total spectral triple with effective Hilbert space {\small{$\mathcal{H}_{\t, f}$}} and relative to an initial spectral triple \isp. Then, given the induced pair {\small{$(\mathcal{B}_{\t}, \mathcal{M}_{\t})$}} the following isomorphism holds: 
$$(\mathcal{C}^{\bullet}_{\Hoch, \Delta}(\mathcal{B}_{\t}, \mathcal{M}_{\t}), d_{\Hoch, \Delta}) \cong (\mathcal{C}^{\bullet}_{\BV}(X_{\t}, d_{S_{\t}}), d_{S_{\t}})$$
for {\small{$\mathcal{C}^{\bullet}_{\Hoch}(\mathcal{B}_{\t}, \mathcal{M}_{\t})$}} the Hochschild complex of the coalgebra {\small{$\mathcal{B}_{\t}$}} over the vomodule {\small{$\mathcal{M}_{\t}$}} and {\small{$\mathcal{C}^{\bullet}_{\BV}(X_{\t}, d_{S_{\t}})$}} the BV complex induced by the total theory \stXS\ encoded in the total spectral triple \TOTsp.
\end{theorem}

\begin{proof}
The proof of this theorem is extremely similar to the one presented for Theorem \ref{Theorem: BV and Hochschild}. Starting by considering the space of cochains, we have that:
$$\mathcal{C}^{n}_{\Hoch, \Delta}(\mathcal{B}_{\t}, \mathcal{M}_{\t}) :=  \mathcal{M}_{\t} \otimes T^{n}(\mathcal{B}_{\t}) \cong \Sym_{\mathcal{M}_{\t}}^{n}(\mathcal{B}_{\t}) =\Sym_{\mathcal{O}_{X_{0}}}^{n}(X_{\t}) =: \mathcal{C}^{n}_{\BV}(X_{\t}, d_{S_{\t}}),$$
in any degree {\small{$n \in \mathbb{Z}$}}. For what concerns the coboundary operators, one should prove that the isomorphism between the cochain spaces is compatible with the two coboundary operators. However, once again one could restrict to only verify this condition on the generators. The statement follows straightforwardly from the explicit expression found for the coproduct {\small{$\Delta_{t}$}}.
\end{proof}

\noindent
This last theorem allows to conclude that, as initially claimed, also the process of further extending a BV-extended theory via the introduction of the required pairs of auxiliary fields coherently fits in the framework provided by the language of finite spectral triples.\\
\\
\noindent
{\bf The gauge-fixing fermion.} As already mentioned, the extension of the BV spectral triple via the introduction of auxiliary fields was just a technical step in order to be able to properly perform a gauge-fixing process by the definition of a suitable gauge-fixing fermion {\small{$\Psi$}}. Hence, to conclude the section, we state the notion of {\emph{gauge-fixing fermion}} within the context of finite spectral triples. 

\begin{definition}
Let {\small{$(\mathcal{A}_{\t}, \mathcal{H}_{\t}, D_{\t}, J_{\t})$}} be a total spectral triple relative to a gauge theory \siXS\ and  with effective Hilbert space 
$$\mathcal{H}_{\t, f} = \mathcal{H}_{\BV, f} \oplus \mathcal{H}_{\aux, f}$$
where 
 $$\mathcal{H}_{\BV, f} = \mathcal{Q}^{*}_{f} [1] \oplus \mathcal{Q}_{f} \quad \mbox{and} \quad \mathcal{H}_{\aux, f} = \mathcal{R}^{*}_{f} [1] \oplus \mathcal{R}_{f} . $$
Then a {\emph{gauge-fixing fermion}} {\small{$\Psi$}} for it is defined as $\Psi \in [\mathcal{O}_{\mathcal{Q}_{f} \oplus \mathcal{R}_{f}}]^{-1},$ with $\Psi\equiv 0.$
\end{definition}

\begin{oss} Before continuing, let us make few remarks about the conditions imposed on a gauge-fixing fermion in the above definition. Explicitly, we are requiring that a gauge-fixing fermion is a regular function of degree $-1$, homotopically equivalent to the zero function and defined on the part of the effective total Hilbert space {\small{$\mathcal{H}_{\t, f}$}} corresponding  to the ghost sector. Asking that a  gauge-fixing fermion {\small{$\Psi$}} is defined only on the ghost sector is trivially enforced by what is the goal of the entire gauge-fixing process, which is eliminating the antifields/anti-ghost sector. For what concerns the condition of its total degree being $-1$, that is a standard condition imposed on the gauge-fixing fermion and due to a degree balancing, for the gauge-fixed action still being of total degree $0$. Finally, we impose to {\small{$\Psi$}} to be homotopically equivalent to the zero function. As consequence, the Lagrangian submanifold defined by the gauge-fixing conditions determined by the gauge-fixing fermion {\small{$\Psi$}} would be homotopically equivalent to the manifold determined by the initial configuration space \siX.
\end{oss}

As recalled in Section \ref{Sect: intro to BV}, once a gauge-fixing has been implemented, still there is a residual symmetry, which turns out to be described by the BRST complex. Thus to claim that the entire BV costruction finds a natural description in terms of spectral triples and within the framework of NCG, also the BRST cohomology complex should have a natural description in terms of a cohomology complex that usually appear within the context of NCG. Even more, in order to have a coherent picture, the BRST complex is expected to coincide once again with a Hochschild complex induced by a total spectral triple considered together with a suitable gauge-fixing fermion {\small{$\Psi$}}. This is precisely what happens: the aim of the next section will be to provide the complete proof of this claim.

\noindent

\section{The BRST cohomology as a Hochschild complex}
\label{Sect: BRST and Hochschild complex}
\noindent
In this section, we present one of the key results of the paper: for the first time we relate the BRST cohomology complex induced by a finite gauge theory to the Hochschild complex determined by a pair {\small{$(\mathcal{B}_{\Psi}, \mathcal{M}_{\Psi})$}}, for {\small{$\mathcal{B}_{\Psi}$}} a coalgebra and {\small{$\mathcal{M}_{\Psi}$}} a comodule on it, all of this within the framework of NCG. Establishing a connection between these two cohomological theories can offer a new point of view to approach the study of the BRST cohomology, a theory which has great relevance from a physical perspective but which still deserves a deeper mathematical investigation. \\
\\
However, relating BRST and Hochschild cohomology will be just a \textquotedblleft side effect\textquotedblright. Indeed, the main purpose of this section is to complete the description of the BV construction using the language of spectral triples. Hence, we are not only facing the problem of relating the BRST complex to another cohomology complex naturally appearing in the context of NCG but also ensuring that the transposition of this last step of the BV construction in the setting of finite spectral triple is coherent and consistent with all the previous steps.  Precisely, as we conclude the previous section by introducing the concept of total spectral triple and determining the suitable transposition of the notion of gauge-fixing fermion in terms  of spectral triples, here what we have to do is to establish how, given a total spectral triple and a suitable gauge-fixing fermion {\small{$\Psi$}} on it, one can construct the BRST complex determined by the gauge-fixed theory {\small{$(X_{\t}, S_{\t})|_{\Psi}$}}, for  {\small{$(X_{\t}, S_{\t})$}} the total theory represented by the total spectral triple \TOTsp. 

 $$\begin{array}{cccccc}
(\mathcal{A}_{\t}, \mathcal{H}_{\t}, D_{\t}, J_{\t}), & \Psi\in [\mathcal{O}_{\mathcal{Q}_{f} \oplus \mathcal{R}_{f}}]^{-1} & \longrightarrow  & \mathcal{C}^{\bullet}_{\Hoch}(\mathcal{B}_{\Psi}, \mathcal{M}_{\Psi}) & \cong & \mathcal{C}^{\bullet}_{\BRST}(X_{\t}|_{\Psi}, d_{S_{\t}}|_{\Psi}) \\
\mbox{\small{total}} & \mbox{\small{gauge-fixing }} & & \mbox{\small{Hochschild}} & & \mbox{\small{BRST cohomology}} \vspace{-1mm}\\
\mbox{\small{spectral triple}} &
\mbox{\small{fermion}} & & \mbox{\small{complex }}& & \mbox{\small{complex }} 
\end{array}
$$

\noindent
As a first attempt, one might be tempted to try to implement the gauge-fixing procedure already at the level of the spectral triple, with the aim of determining a sort of gauge-fixed version of a total spectral triple. 
Following a procedure similar to the one presented in Section \ref{Sect: BV spectral triple} and \ref{Sect: gauge-fixing procedure in terms of NCG}, one might try to define a gauge-fixed spectral triple to be a finite (eventually mixed KO-dimensional) spectral triple satisfying conditions which are analogous to the ones required for the BV and the total spectral triple. In particular, first of all one would impose that this gauge-fixed spectral triple encodes the gauge-fixed pair {\small{$(X_{\t}, S_{\t})|_{\Psi}$}}, as the BV spectral triple and the total spectral triple did encode the BV-extended theory \swXS\ and the total theory {\small{$(X_{\t}, S_{\t})$}}, respectively. Moreover, one would also require that the  Hochschild cohomology complex induced by such gauge-fixed spectral triple coincides with the BRST complex determined by the pair {\small{$(X_{\t}, S_{\t})|_{\Psi}$}}. \\
\\
However, a comparison with how the construction is classically performed would suggest that the restriction to the ghost sector can occur only at the level of the cohomology complex. As briefly recalled in Section \ref{Subsection: The BRST cohomology complex}, given a BV complex, one obtains the corresponding BRST complex by imposing the gauge-fixing condition both on the cochain spaces as well as on the action of the coboundary operator. In other words, strictly speaking, the gauge-fixing process is not applied to the pair \stXS\ but it is actually performed only after having constructed the cohomology complex. Similarly, also in this spectral triple setting we expect the gauge-fixing procedure not to  act st the spectral triple level but to be performed only after having determined the pair {\small{$(\mathcal{B}_{\t},  \mathcal{M}_{\t})$}} induced by a total spectral triple. As anticipated, to have a coherent construction, the way of performing such a gauge-fixing process on  {\small{$(\mathcal{B}_{\t},  \mathcal{M}_{\t})$}} should lead to the existence of an isomorphism between the Hochschild complex defined by the gauge-fixed pair {\small{$(\mathcal{B}_{\Psi}, \mathcal{M}_{\Psi})$}} and the BRST complex of the gauge-fixed theory {\small{$(X_{\t},  S_{\t})|_{\Psi}$}}.\\
\\
To further reinforce the intuition that one cannot take the total spectral triple as starting point of the construction but rather the gauge-fixing procedure has to be implemented directly on the induced pair {\small{$(\mathcal{B}_{\t}, \mathcal{M}_{\t})$}} one could observe that, by enforcing the gauge-fixing condition on the pair \stXS, we completely change its mathematical structure. Consequently, it is not reasonable to expect {\small{$(X_{\t}, S_{\t})|_{\Psi}$}} to be nonetheless encoded in the structure of a spectral triple as done for the pairs \swXS\ and \stXS. Therefore the translation of the gauge-fixing procedure to the context of NCG would actually yields to determining how to perform the following process:
 $$\begin{array}{cccccc}
(\mathcal{B}_{\t}, \mathcal{M}_{\t}), & \Psi\in [\mathcal{O}_{\mathcal{Q}_{f} \oplus \mathcal{R}_{f}}]^{-1} & \vector(1,0){30}  & (\mathcal{B}_{\Psi}, \mathcal{M}_{\Psi})\\
\mbox{\small{pair induced by}} & \mbox{\small{gauge-fixing }} & & \mbox{\small{gauge-fixed}}  \vspace{-1mm}\\
\mbox{\small{the total spectral triple}} & \mbox{\small{fermion}} && \mbox{\small{pair}}
\end{array}
$$
\noindent
{\bf The algebra $\mathcal{B}_{\Psi}$.} \ Starting with the algebra {\small{$\mathcal{B}_{\t}$}}, the goal we want to reach by implementing on it a gauge-fixing process is, as usual, to eliminate all the antifields/anti-ghost fields from it. Therefore, we define {\small{$\mathcal{B}_{\Psi}$}} to be simply
\begin{equation}
    \label{Eq: def B Psi}
\B_{\Psi, 0} = [\mathcal{O}_{X_{0}}]^{\leq(\deg(S_{0})-1)} \oplus [i \mathfrak{su}(n)], \quad \B_{\Psi,-1} = [\mathfrak{su}(n)], \quad \mbox{and} \quad  \B_{\Psi,m} = [\mathcal{Q}_{f}]^{m}, 
\end{equation}
for {\small{$m>1$}}. In other words, instead  of having an entire collection of free variables in {\small{$X^{*}_{0}[1]$}}, {\small{$\mathcal{Q}_{f}^{*}[1]$}} and {\small{$\mathcal{R}_{f}^{*}[1]$}}, we force them to lay on the Lagrangian submanifold determined by the so-called gauge-fixing condition {\small{$\{\varphi^{*}_{j} =\partial \Psi/\partial \varphi_{j}\}$}} and {\small{$\{\chi^{*}_{k} =\partial \Psi/\partial \chi_{k}\}$}}, where by {\small{$\varphi_{j}$}} and {\small{$\chi_{k}$}} we denote, respectively, a generic field/ghost field in {\small{$\mathcal{Q}_{f}$}} and {\small{$\mathcal{R}_{f}$}}. Therefore, in this context and under these conditions {\small{$\mathcal{B}_{\Psi}$}} is simply obtained as restriction of {\small{$\mathcal{B}_{\t}$}} to its ghost sector. Finally, for what concerns its coproduct structure, we remark that, due to the linearity of {\small{$S_{\t}$}} in the antifields/anti-ghost fields, then
$$\Delta_{\Psi}(y^{a})= \Delta_{\t}(y^{a}),$$
for any generator {\small{$y^{a} \in \B_{\Psi}$}}.\\
\\
\noindent
{\bf The module $\mathcal{M}_{\Psi}$.} \ Similarly to what done for the algebra {\small{$\mathcal{B}_{\Psi}$}}, also the module {\small{$\mathcal{M}_{\Psi}$}} is obtained as restriction of {\small{$\mathcal{M}_{\t}$}} to the ghost sector. However, because 
$$\mathcal{M}_{\t}:= \langle \Omega^{1}(\mathcal{A}_{\t})\rangle  =\langle \Omega^{1}(\mathcal{A}_{0})\rangle  $$
is the  module generated by the initial fields in \siX, performing the gauge-fixing process on it does not have any effect as there are no antifields/anti-ghost fields to eliminate. Hence:
\begin{equation}
\label{eq: M_Psi}
\mathcal{M}_{\Psi} = \langle \Omega^{1}(\mathcal{A}_{0})\rangle|_{\Psi} = \langle \Omega^{1}(\mathcal{A}_{0})\rangle = \mathcal{M}.
\end{equation}

\noindent
Having defined the gauge-fixed pair {\small{$(\mathcal{B}_{\Psi}, \mathcal{M}_{\Psi})$}} we can finally state the main result of this section.

\begin{theorem}
Given {\small{$(\mathcal{A}_{\t}, \mathcal{H}_{\t}, D_{\t}, J_{\t})$}} the total spectral triple associated to an initial spectral triple {\small{$(\mathcal{A}_{0}, \mathcal{H}_{0}, D_{0})$}} and {\small{$\Psi \in [\mathcal{O}_{\mathcal{Q}_{f} \oplus \mathcal{R}_{f}}]^{-1}$}} a gauge-fixing fermion for the theory, let {\small{$(\mathcal{B}_{\Psi}, \mathcal{M}_{\Psi})$}} be the pair defined in \eqref{Eq: def B Psi} and \eqref{eq: M_Psi}. Then, the following isomorphism holds
$$\mathcal{C}^{\bullet}_{\Hoch, \Delta}(\mathcal{B}_{\Psi}, \mathcal{M}_{\Psi}) \cong \mathcal{C}^{\bullet}_{\BRST}(X_{\t}, d_{S_{\t}})|_{\Psi}$$
where {\small{$\mathcal{C}^{\bullet}_{\BRST}(X_{\t}, d_{S_{\t}})|_{\Psi}$}} denotes the BRST complex associated to the initial gauge theory \siXS\ induced by the initial spectral triple  {\small{$(\mathcal{A}_{0}, \mathcal{H}_{0}, D_{0})$}}. 
\end{theorem}

\begin{proof}
The statement follows straightforwardly from noticing that both the complexes considered can be viewed as restriction to the ghost sector of the corresponding complexes defined on the total spectral triple. As {\small{$\mathcal{C}^{\bullet}_{\Hoch, \Delta}(\mathcal{B}_{\t}, \mathcal{M}_{\t}$}} and {\small{$\mathcal{C}^{\bullet}_{\BV}(X_{\t}, d_{S_{\t}}$}} have been proved to be isomorphic in Theorem \ref{Theo: BV tot}, the isomorohism of {\small{$\mathcal{C}^{\bullet}_{\Hoch}(\mathcal{B}_{\Psi}, \mathcal{M}_{\Psi})$}} and {\small{$\mathcal{C}^{\bullet}_{\BRST}(X_{\t}, d_{S_{\t}})|_{\Psi}$}} follows.
\end{proof}

\noindent
\begin{figure}[h]
\centering
\includegraphics[width=0.99\textwidth]{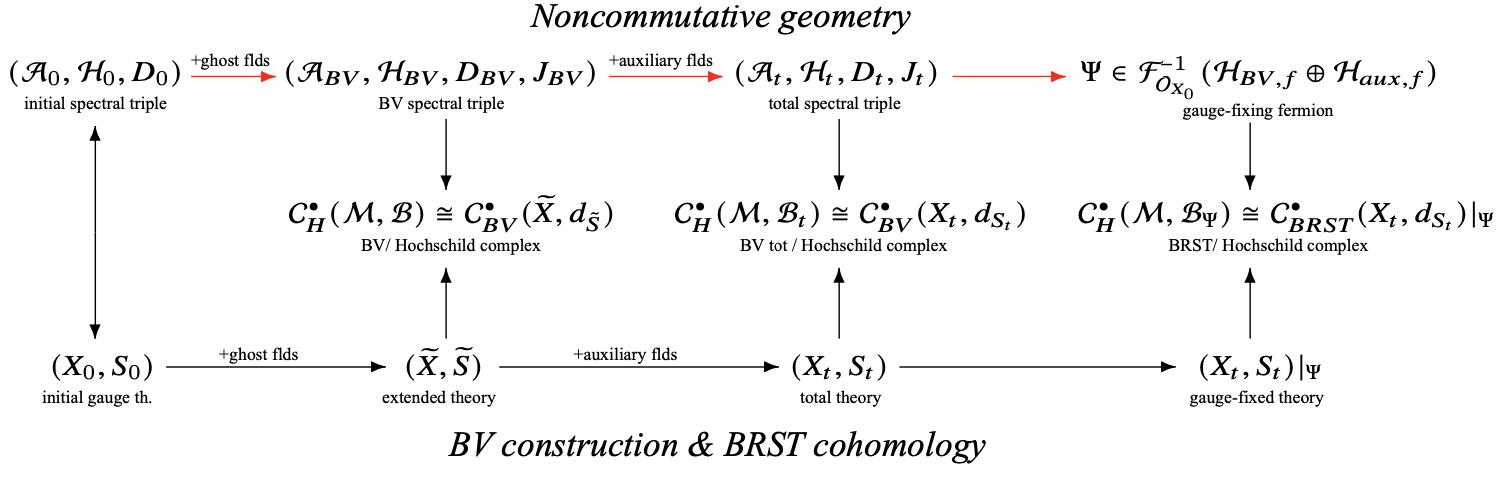}
\caption{The BV construction for finite spectral triples and its compatibility at the level of the cohomology complexes}
\label{Fig: Complete construction}
\end{figure}

\section{Conclusions: the BV construction in the setting of NCG}
\label{Sect: conclusion}
\noindent
In this article we proved how the whole BV construction, from the introduction of ghost/anti-ghost fields to the construction of the BV complex, can be performed using the language of spectral triple. Grafically, this construction has been summarized in Figure \ref{Fig: Complete construction}. Even tough here we have been focusing on finite dimensional gauge theories and hence on finite spectral triples, the extension of the construction to the case of gauge theories induced by almost commutative spectral triple appears to be natural and it is currently under investigation. The motivation and the relevance of performing this extra step lies in the fact that the class of theories induced by almost commutative spectral triples includes models that are extremely relevant from a physical point of view, such as the Standard Model of particle physics.

\printbibliography
\end{document}